%% file: paper.tex
\documentclass{amsart}
\pdfoutput=1
\usepackage{amssymb}
\usepackage{hyperref}
\usepackage{todonotes}

\newtheorem{Lem}{Lemma}
\newtheorem{Thm}{Theorem}

\begin{document}


\title%
[]%
{Counting thin subgraphs via packings\\faster than meet-in-the-middle time}
\author%
[]%
{Andreas Bj\"orklund and Petteri Kaski and \L ukasz Kowalik}

\begin{abstract}
Vassilevska and Williams~(STOC 2009) showed how to count simple paths on $k$ vertices and matchings on $k/2$ edges in an $n$-vertex graph in time $n^{k/2+O(1)}$. In the same year, two different algorithms with the same runtime were given by Koutis and Williams~(ICALP 2009), and Bj\"orklund \emph{et al.}~(ESA 2009), via $n^{st/2+O(1)}$-time algorithms for counting $t$-tuples of pairwise disjoint sets drawn from a given family of $s$-sized subsets of an $n$-element universe.
Shortly afterwards, Alon and Gutner~(TALG 2010) showed that these problems have $\Omega(n^{\lfloor st/2\rfloor})$ and $\Omega(n^{\lfloor k/2\rfloor})$ lower bounds when counting by color coding. 

Here we show that one can do better, namely, we show that the ``meet-in-the-middle'' exponent $st/2$ can be beaten and give an algorithm that counts in time $n^{0.45470382 st + O(1)}$ for $t$ a multiple of three. This implies algorithms for counting occurrences of a fixed subgraph on $k$ vertices and pathwidth $p\ll k$ in an $n$-vertex graph in $n^{0.45470382k+2p+O(1)}$ time, improving on the three mentioned algorithms for paths and matchings, and circumventing the color-coding lower bound. 
We also give improved bounds for counting $t$-tuples of disjoint $s$-sets for $s=2,3,4$. 

Our algorithms use fast matrix multiplication. We show an argument that this is necessary to go below the meet-in-the-middle barrier.
\end{abstract}

\maketitle


\section{Introduction}

Suppose we want to count the number of occurrences of a $k$-element pattern 
in an $n$-element universe. This setting is encountered, for example,
when $P$ is a $k$-vertex pattern graph, $H$ is an $n$-vertex host graph, and 
we want to count the number of subgraphs that are isomorphic to $P$ in $H$.
If $k$ is a constant independent of $n$, enumerating all the $k$-element 
subsets or tuples of the $n$-element universe can be done in time $O(n^k)$, 
which presents a trivial upper bound for counting small patterns. 

In this paper we are interested in patterns that are {\em thin}, such 
as pattern graphs that are paths or cycles, or more generally pattern graphs 
with {\em bounded pathwidth}. Characteristic to such patterns is 
that they can be split into two or more parts, such that the interface 
between the parts is easy to control. For example, a simple path on $k$ 
vertices can be split into two paths of half the length that have exactly 
one vertex in common; alternatively, one may split the path into 
two independent sets of vertices.

The possibility to split into two controllable parts immediately suggests 
that one should pursue an algorithm that runs in no worse 
time than $n^{k/2+O(1)}$; such an algorithm was indeed 
discovered in 2009 by Vassilevska and Williams~\cite{VW09} for 
counting $k$-vertex subgraphs that admit an independent set of size $k/2$.
This result was accompanied, within the same year, of two publications 
presenting the same runtime restricted to counting paths and matchings. Koutis 
and Williams~\cite{KW09} and Bj\"orklund \emph{et al.}~\cite{BHKK09} 
describe different algorithms for the related problem of counting 
the number of $t$-tuples of disjoint sets that can be formed from
a given family of $s$-subsets of an $n$-element universe in 
$n^{st/2+O(1)}$ time. Fomin~\emph{et al.}~\cite{FLRRS12} 
generalized the latter result into an algorithm that counts occurrences 
of a $k$-vertex pattern graph with pathwidth $p$ in $n^{k/2+2p+O(1)}$ time.

Splitting into three parts enables faster listing of the parts in 
$n^{k/3+O(1)}$ time, but requires more elaborate control at the interface 
between parts. This strategy enables one to count also dense subgraphs
such as $k$-cliques via an algorithm of Ne\v{s}et\v{r}il and 
Poljak \cite{NP85} (see also \cite{EG04,KKM00}) that uses fast
matrix multiplication to achieve a pairwise join of the three parts, 
resulting in running time $n^{\omega k/3+O(1)}$, where 
$2\leq\omega< 2.3728639$ is the limiting exponent of square matrix 
multiplication \cite{LG14,VW12}.
Even in the case $\omega=2$ this running time is, however, 
$n^{2k/3+O(1)}$, which is inferior to ``meeting in the middle'' by 
splitting into two parts.

But is meet-in-the-middle really the best one can do?
For many problems it appears indeed that the worst-case running time given 
by meet-in-the-middle is difficult to beat. Among the most notorious examples 
in this regard is the Subset Sum problem, for which the 1974 
meet-in-the-middle algorithm of Horowitz and Sahni~\cite{HS74} remains to 
date the uncontested champion. Related problems such as the $k$-Sum problem 
have an equally frustrating status, in fact to such an extent that the 
case $k=3$ is regularly used as a source of hardness reductions in 
computational geometry~\cite{GO12}.

Against this background one could perhaps expect a barrier 
at the meet-in-the-middle time $n^{k/2+O(1)}$ for counting
thin subgraphs, and such a position would not be without some supporting 
evidence. Indeed, not only are the algorithms of Vassilevska and 
Williams~\cite{VW09}, Koutis and Williams~\cite{KW09}, and Bj\"orklund \emph{et al.}~\cite{BHKK09} fairly recent discoveries, but they all employ 
rather different techniques. Common to all three algorithms is however 
the need to consider the $k/2$-element subsets of the $n$-element vertex set, 
resulting in time $n^{k/2+O(1)}$. Yet further evidence towards a barrier 
was obtained by Alon and Gutner~\cite{AG09} who showed that 
color-coding based counting approaches relying on 
a perfectly $k$-balanced family of hash functions face 
an unconditional $c(k)n^{\lfloor k/2\rfloor}$ lower bound for the size 
of such a family. From a structural complexity perspective 
Flum and Grohe~\cite{FG04} have shown that counting $k$-paths is \#W[1]-hard 
with respect to the parameter $k$, and a very recent breakthrough of 
Curticapean~\cite{C13} establishes a similar barrier to counting $k$-matchings.
This means that parameterized counting algorithms with running time $f(k)n^{O(1)}$ 
for a function $f(k)$ independent of $n$ are unlikely for these problems,
even if such a structural complexity approach does not pinpoint precise
lower bounds of the form $n^{\Omega(g(k))}$ for some function $g(k)$. 

Contrary to the partial evidence above, however, our objective
in this paper is to show that 
{\em there is a crack in the meet-in-the-middle barrier}, 
albeit a modest one. In particular, we show that it is possible 
to count subgraphs on $k$ vertices such as paths and matchings---and more 
generally any $k$-vertex subgraphs with pathwidth $p$---within 
time $n^{0.45470382k+2p+O(1)}$ for $p\ll k$.

Our strategy is to reduce the counting problem to the task of 
evaluating a particular trilinear form on weighted hypergraphs, and 
then show that this trilinear form admits an evaluation algorithm that breaks 
the meet-in-the-middle barrier. This latter algorithm is our main contribution,
which we now proceed to present in more detail.

\subsection{Weighted disjoint triples}

Let $U$ be an $n$-element set. For a nonnegative integer $q$, let us write $\binom{U}{q}$ for the set of all $q$-element subsets of $U$. 
Let $f,g,h:\binom{U}{q}\rightarrow \mathbb{Z}$ be three functions given 
as input. We are interested in computing the trilinear form
\begin{equation}
\label{eq:main}
\Delta(f,g,h)=\!\!\!\!\!\!\!\!\!\sum_{\substack{A,B,C\in\binom{U}{q}\\A\cap B=A\cap C=B\cap C=\emptyset}}\!\!\!\!\!\!\!\!\!
f(A)g(B)h(C)\,.
\end{equation}
To ease the running time analysis, we make two assumptions. 
First, $q$ is a constant independent of $n$. 
Second, we assume that the values of the functions $f,g,h$ are 
bounded in bit-length by a polynomial in $n$, which will be the setup 
in our applications (Theorems~\ref{thm:subgraph}~and~\ref{thm:packings}).

Let us write $\omega$ for the limiting exponent of square matrix 
multiplication, $2\leq \omega<2.3728639$ \cite{LG14,VW12}.
Similarly, let us write $\alpha$ for the limiting 
exponent such that multiplying an $N\times N^\alpha$ matrix with 
an $N^\alpha\times N$ matrix takes $N^{2+o(1)}$ arithmetic operations, 
$0.3<\alpha\leq 3-\omega$~\cite{LG12}. 

The next theorem is our main result; here the intuition is that we
take $k=3q$ in our applications, implying that we break the 
meet-in-the-middle exponent $k/2$.

\begin{Thm}[Fast weighted disjoint triples]
\label{thm:main}
There exists an algorithm that evaluates $\Delta(f,g,h)$ 
in time $O\bigl(n^{3q(\frac{1}{2}-\tau)+c}\bigl)$
for constants $c$ and $\tau$ independent of the constant $q$,
with $c\geq 0$ and 
\begin{equation}
\label{eq:tau}
\tau=
\begin{cases}
\frac{(3-\omega)(1-\alpha)}{36-6(1+\omega)(1+\alpha)} & \text{if $\alpha\leq 1/2$};\\
\frac{1}{18} & \text{if $\alpha\geq 1/2$.}
\end{cases}
\end{equation}
\end{Thm}

\noindent
{\em Remark 1.} 
For $\omega=2.3728639$ and $\alpha=0.30$ we obtain $\tau=0.045296182$ and hence
$O\bigl(n^{3q\cdot 0.45470382+c}\bigr)$ time.
For $\alpha\geq 1/2$ we obtain $\tau=0.055555556$ and hence
$O\bigl(n^{3q\cdot 0.44444445+c}\bigr)$ time. Note that the latter case
occurs in the case $\omega=2$ because then $\alpha=1$.

\medskip
\noindent
{\em Remark 2.} 
We observe that the trilinear form \eqref{eq:main} admits an evaluation 
algorithm analogous to the algorithm of 
Ne\v{s}et\v{r}il and Poljak \cite{NP85} discussed above. 
Indeed, \eqref{eq:main} can be split into a multiplication of two
$n^{q}\times n^{q}$ square matrices, which gives running time
$O(n^{\omega q+c})$. Even in the case $\omega=2$ the running
time $O(n^{2q+c})$ is however inferior to Theorem~\ref{thm:main}.

\medskip
\noindent
{\em Remark 3.} 
Theorem~\ref{thm:main} can be stated in an alternative form that
counts the number of arithmetic operations (addition, subtraction, 
multiplication, and exact division of integers) performed by the 
algorithm on the inputs $f,g,h$ to obtain $\Delta(f,g,h)$. 
This form is obtained by simply removing the constant $c$ from 
the bound in Theorem~\ref{thm:main}. 

Finally, we show that one can improve upon Theorem~\ref{thm:main}
via case by case analysis. Here our intent is to pursue only the 
cases $q=2,3,4$ and leave the task of generalizing from here to
further work. 

When considering specific values of $q$, it is convenient to 
measure efficiency using the number of arithmetic operations 
(addition, subtraction, multiplication, and exact division of 
integers) performed by an algorithm. 

\begin{Thm}
\label{thm:q234}
There exist algorithms that solve the weighted disjoint triples 
problem 
\begin{enumerate}
\item
for $q=2$ in $O(n^{\omega})$ arithmetic operations,
\item
for $q=3$ in $O(n^{\omega+1})$ arithmetic operations, and
\item
for $q=4$ in $O(n^{2\omega})$ arithmetic operations.
\end{enumerate}
\end{Thm}

{\em Remark.} In the case $\omega=2$ we observe that the three
algorithms in Theorem~\ref{thm:q234} all run in $O(n^q)$ 
arithmetic operations, which is linear in the size of the input.

\subsection{Counting thin subgraphs and packings}

Once Theorem~\ref{thm:main} is available, the following theorem is an 
almost immediate corollary of techniques for counting injective 
homomorphisms of bounded-pathwidth graphs developed by 
Fomin~{\em et al.}~\cite{FLRRS12} 
(see also \S3 in Amini~{\em et al.}~\cite{AFS12}).
In what follows $\tau$ is the constant in \eqref{eq:tau}. 

\begin{Thm}[Fast counting of thin subgraphs]
\label{thm:subgraph}
Let $P$ be a fixed pattern graph with $k$ vertices and pathwidth $p$. 
Then, there exists an algorithm that takes as input an $n$-vertex
host graph $H$ and counts the number of subgraphs of $H$ that are 
isomorphic to $P$ in time 
$O\bigl(n^{(\frac{1}{2}-\tau)k+2p+c}+n^{k/3+3p+c}\bigr)$
where $c\geq 0$ is a constant independent of the constants $k,p,\tau$.
\end{Thm}

\noindent
{\em Remark.} The running time in Theorem~\ref{thm:subgraph} 
simplifies to $O\bigl(n^{(\frac{1}{2}-\tau)k+2p+c}\bigr)$ if $p\leq k/9$.

\medskip
Theorem~\ref{thm:main} gives also an immediate speedup for counting set 
packings. In this case we use standard dynamic programming to count, 
for each $q$-subset $A$ with $q=st/3$, the number of $t/3$-tuples of 
pairwise disjoint $s$-subsets whose union is $A$. 
We then use Theorem~\ref{thm:main} to 
assemble the number of $t$-tuples of pairwise disjoint $s$-subsets
from triples of such $q$-subsets. This results in the following corollary.

\begin{Thm}[Fast counting of set packings]
\label{thm:packings}
There exists an algorithm that takes as input 
a family $\mathcal{F}$ of $s$-element subsets of an $n$-element 
set and an integer $t$ that is divisible by $3$, 
and counts the number of $t$-tuples of pairwise disjoint subsets
from $\mathcal{F}$ in time $O\bigl(n^{(\frac{1}{2}-\tau)st+c}\bigr)$
where $c\geq 0$ is a constant independent of the constants $s,t,\tau$.
\end{Thm}

\subsection{On the hardness of counting in disjoint parts}

We present two results that provide partial
justification why there was an apparent barrier 
at ``meet-in-the-middle time'' for counting in disjoint parts. 

First, in the case of two disjoint parts, the problem appears to
contain no algebraic dependency that one could expoit towards
faster algorithms beyond those already presented in 
Bj\"orklund {\em et al.}~\cite{BHKK08,BHKK09}. Indeed, we can
provide some support towards this intuition by showing that
the associated 2-tensor has full rank over the rationals, see Lemma~\ref{lem:disjmat}.
This observation is most likely not new but we were unable to find the right reference.

Second, recall that our algorithms mentioned in the previous section use fast matrix multiplication. 
We show an argument that this is necessary to go below the meet-in-the-middle barrier.
More precisely, we show that any {\em trilinear algorithm} (cf.~\cite[\S9]{Pan1984})
for $\Delta(f,g,h)$ whose rank over the integers is below the 
meet-in-the-middle barrier implies a sub-cubic algorithm for matrix 
multiplication:

\begin{Thm}
\label{thm:omega-tau}
Suppose that for all constants $q$ there exists
a trilinear algorithm for $\Delta(f,g,h)$ with 
rank $r=O(n^{3q(1/2-\tau)+c})$ over the integers,
where $\tau>0$ and $c\geq 0$ are constants 
independent of $n$ and $q$. 
Then, $\omega\leq 3-\tau$.
\end{Thm}

\subsection{Overview of techniques and discussion}

The main idea underlying Theorem~\ref{thm:main} is to design
a system of linear equations whose solution contains the weighted 
disjoint triples \eqref{eq:main} as one indeterminate. The main obstacle
to such a design is of course that we must be able to construct 
and solve the system within the allocated time budget.

In our case the design will essentially be a balance between
two families of linear equations, the {\em basic} (first) family and
the {\em cheap} (second) family, for the same indeterminates. The basic 
equations alone suffice to solve 
the system in meet-in-the-middle time $O(n^{3q/2+c})$, whereas
the cheap equations solve directly for selected indeterminates 
{\em other than} \eqref{eq:main}. The virtue of the cheap equations 
is that their right-hand sides can be evaluated efficiently using 
fast (rectangular) matrix multiplication, which enables us to throw 
away the most expensive of the basic equations and still have sufficient 
equations to solve for \eqref{eq:main}, thereby breaking the 
meet-in-the-middle barrier. Alternatively one can view the 
extra indeterminates and linear equations as a tool to expand
the scope of our techniques beyond the extent of the apparent 
barrier so that it can be circumvented.

Before we proceed to outline the design in more detail, let us observe
that the general ingredients outlined above, namely fast matrix multiplication
and linear equations, are well-known techniques employed in a number of
earlier studies. In particular in the context of subgraph counting 
such techniques can be traced back at least to the triangle- and 
cycle-counting algorithms of Itai and Rodeh \cite{IR78}, with more recent 
uses including the algorithms of Kowaluk, Lingas, and Lundell~\cite{KLL}
that improve upon algorithms of 
Ne\v{s}et\v{r}il and Poljak \cite{NP85}
and 
Vassilevska and Williams \cite{VW09} 
for counting small dense subgraphs 
($k<10$) with a maximum independent set of size $2$. 
Also the counting-in-halves technique of 
Bj\"orklund {\em et al.}~\cite{BHKK09} can be seen to solve an 
(implicit) system of linear equations to recover weighted disjoint 
packings.

Let us now proceed to more detailed design considerations.
Here the main task is to relax \eqref{eq:main} 
into a collection of trilinear forms related by linear constraints.
A natural framework for relaxation is to parameterize the triples 
$(A,B,C)$ so that the pairwise disjoint triples required by \eqref{eq:main} 
become an extremal case. 

A first attempt at such parameterization is to parameterize the triples 
$(A,B,C)$ by the size of the union $|A\cup B\cup C|=j$. In particular,
the triple $(A,B,C)$ is pairwise disjoint if and only if $j=3q$.
With this parameterization we obtain $2q+1$ indeterminates, one for each 
value of $q\leq j\leq 3q$. In this case {\em inclusion-sieving} 
(trimmed M\"obius inversion \cite{BHKK09,BHKK10}) 
on the subset lattice $(2^{[n]},\cup)$ enables a system of linear
equations on the indeterminates. This is in fact the approach underlying 
the counting-in-halves technique of Bj\"orklund {\em et al.}~\cite{BHKK09},
which generalizes also to M\"obius algebras of lattices with 
the set union (set intersection) replaced by the join (meet) operation of 
the lattice~\cite{BHKKNP12}. Unfortunately, it appears difficult to break 
the meet-in-the-middle barrier via this parameterization, in particular 
due to an apparent difficulty of arriving at a cheap system of equations 
to complement the basic equations arising from the inclusion sieve.

A second attempt at parameterization is to replace the set union $X\cup Y$ 
with the {\em symmetric difference} 
$X\oplus Y=(X\setminus Y)\cup(Y\setminus X)$ and parameterize
the triples $(A,B,C)$ by the size of the symmetric difference 
$|A\oplus B\oplus C|=j$. The set $A\oplus B\oplus C$ is illustrated
in Fig.~\ref{fig:type}.

\begin{figure}[ht]
\begin{center}
\def\svgwidth{4cm} 
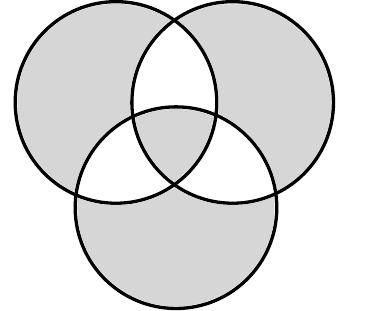 
\end{center}
\vspace*{-4mm}
\caption{The set $A\oplus B\oplus C$.}
\label{fig:type}
\end{figure}

Recalling that $X\cup Y$ and $X\oplus Y$ coincide if and only if $X$ and $Y$ are disjoint, we again 
recover the pairwise disjoint triples as the extremal case $j=3q$. 
With this parameterization we obtain $\lfloor 3q/2\rfloor+1$ indeterminates,
one for each $0\leq j\leq 3q$ such that $j\equiv q\pmod 2$.
In this case {\em parity-sieving} 
(trimmed ``parity-of-intersection transforms'', 
see~\S\ref{sect:parity-trans}) on the group algebra of the 
elementary Abelian group $(2^{[n]},\oplus)$ enables a system of linear 
equations on the indeterminates. While this second parameterization via 
symmetric difference is {\em a priori} less natural than the 
first parameterization via set union, it turns out to be more successful 
in breaking the meet-in-the-middle barrier. In particular the basic 
equations (Lemma~\ref{lem:first}) on the $\lfloor 3q/2\rfloor+1$ indeterminates
alone suffice to obtain an algorithm with running time $O(n^{3q/2+c})$,
which is precisely at the meet-in-the-middle barrier.
The key insight then to break the barrier is that the indeterminates with 
{\em small} values of $j$ can be solved directly (Lemma~\ref{lem:second}) via 
fast rectangular matrix multiplication. In particular this is because small 
$j$ implies large overlap between the sets $A,B,C$ and a 
``triangle-like'' structure that is amenable to matrix multiplication 
techniques. That is, from the perspective of the symmetric difference 
$D=A\oplus B$, it suffices to control the differences $A\setminus B$ and 
$B\setminus A$ (outer dimensions in matrix multiplication), whereas the 
overlap $A\cap B$ (inner dimension) is free to range across sets disjoint 
from $D$ (see~\S\ref{sect:sdp}). 

A further design constraint is that the basic equations (Lemma~\ref{lem:first})
must be mutually independent of the cheap equations (Lemma~\ref{lem:second})
to enable balancing the running time (see~\S\ref{sect:runtime}) 
while retaining invertibility of the system; here we have opted for an 
analytically convenient design where the basic equations are in general 
position (the coefficient matrix is a Vandermonde matrix) that enables easy 
combination with the cheap equations, even though this design may not be 
the most efficient possible from a computational perspective. 

From an efficiency perspective we can in fact do better than 
Theorem~\ref{thm:main} for small values of $q$ by proceeding 
via case by case analysis (see~\S\ref{sect:q234}). 
We show that faster algorithms exist for at least $q=2,3,4$ 
(Theorem~\ref{thm:q234}). 

An open problem that thus remains is whether the upper bound 
$n^{3q(\frac{1}{2}-\tau)+O(1)}$ in Theorem~\ref{thm:main}
can be improved to the asymptotic form 
$\binom{n}{3q(\frac{1}{2}-\delta)}n^{O(1)}$ for some constant $\delta>0$
independent of $0\leq q\leq n/3$. In particular, such an improvement would
parallel the asymptotic running time $\binom{n}{k/2}n^{O(1)}$ of the
counting-in-halves technique~\cite{BHKK09}. Furthermore, such an
improvement would be of considerable interest since it would, for example, 
lead to faster algorithms for computing the permanent of an 
integer matrix. Unfortunately, this also suggests that such an improvement
is unlikely, or at least difficult to obtain given the relatively 
modest progress in improved algorithms for the permament 
problem~\cite{Bjorklund12}. Some further evidence towards
the subtlety of counting in disjoint parts is that we can show
(Theorem~\ref{thm:omega-tau}) that to break the ``meet-in-the-middle'' 
barrier for the weighted disjoint triples problem with a trilinear 
algorithm, it is in fact {\em necessary} to use fast matrix multiplication 
(see~\S\ref{sect:lower-bounds}). Put otherwise, the proofs of
Theorems~\ref{thm:main} and \ref{thm:omega-tau} reveal that 
for constant $q$ the structural tensors for 
weighted disjoint triples and matrix multiplication are
loosely rank-equivalent in terms of existence of low-rank decompositions.

\subsection{Organization}

The proof of Theorem~\ref{thm:main} is split into two parts.
First, in \S\ref{sect:system} we derive a linear system whose
solution contains $\Delta(f,g,h)$. Then, in \S\ref{sect:solving}
we derive an algorithm that constructs and solves the system within the claimed
running time bound. We then proceed with the two highlighted
applications of Theorem~\ref{thm:main}: 
in \S\ref{sect:homomorphims-in-three-parts}
we give a proof of Theorem~\ref{thm:subgraph} by relying on techniques
of Fomin~{\em et al.}~\cite{FLRRS12}, and 
in \S\ref{sect:packings-in-three-parts} we prove Theorem~\ref{thm:packings}.
We conclude the paper in \S\ref{sect:lower-bounds} by connecting
fast trilinear algorithms for $\Delta(f,g,h)$ to fast matrix multiplication.

\section{The linear system}
\label{sect:system}

We now proceed to derive a linear system whose solution contains
$\Delta(f,g,h)$. Towards this end it is convenient to start by
recalling some elementary properties of the symmetric difference 
operator on sets.

For sets $X,Y\subseteq U$, let us write 
$X\oplus Y=(X\setminus Y)\cup(Y\setminus X)$ for the symmetric 
difference of $X$ and $Y$. We immediately observe that 
\begin{equation}
\label{eq:xor-size}
|X\oplus Y|=|X|+|Y|-2|X\cap Y|
\end{equation}
and hence
\begin{equation}
\label{eq:xor-size-equiv-sum}
|X\oplus Y|\equiv |X|+|Y|\pmod 2\,.
\end{equation}
In particular, for any $A,B,C\in\binom{U}{q}$ we have 
\[
|A\oplus B\oplus C|\equiv |A|+|B|+|C|=3q\equiv q\pmod 2\,.
\]
Thus, the size $|A\oplus B\oplus C|$ is always even if $q$ is even 
and always odd if $q$ is odd. In both cases $j=|A\oplus B\oplus C|$
may assume exactly $e=\lfloor 3q/2\rfloor+1$ values in
\begin{equation}
\label{eq:jq}
J_q=\{j\in\{0,1,\ldots,3q\}:j\equiv q\!\!\!\pmod 2\}\,.
\end{equation}

We are now ready to define the $e\times e$ linear system. 
We start with the indeterminates of the system.

\subsection{The indeterminates}
For each $j\in J_q$, let 
\begin{equation}
\label{eq:xj}
x_j=x_j(f,g,h)=
\!\!\!\sum_{\substack{A,B,C\in\binom{U}{q}\\|A\oplus B\oplus C|=j}}\!\!\!
f(A)g(B)h(C)\,.
\end{equation}
In particular, since $A,B,C\in\binom{U}{q}$ are pairwise 
disjoint if and only if $|A\oplus B\oplus C|=3q$, 
we observe that $\Delta(f,g,h)=x_{3q}$. 
Thus, it suffices to solve for the indeterminate $x_{3q}$ 
to recover \eqref{eq:main}. We proceed to formulate a linear
system towards this end. The system is based on two families
of equations. The first family will contribute $d$ equations,
and the second family will contribute $e-d$ equations. 

\subsection{A first family of equations}
Our first family of equations is based on a parity construction.
For now we will be content in simply defining the equations
and providing an illustration in Fig.~\ref{fig:triple-intersect}.
(The eventual algorithmic serendipity of this construction will be
revealed only later in \eqref{eq:xor-cap} and \eqref{eq:tpfgh}.)
Let $i=0,1,\ldots,d-1$ be an index for the equations,
let $p\in\{0,1\}$ denote parity, and let $s=0,1,\ldots,i$.
For all $Z\in\binom{U}{s}$ let
\begin{equation}
\label{eq:tp}
T_p(Z)=
\!\!\!\!\!\!\!\!\!\!\!\!\sum_{\substack{A,B,C\in\binom{U}{q}\\|(A\oplus B\oplus C)\cap Z|\equiv p\!\!\!\pmod 2}}\!\!\!\!\!\!\!\!\!\!\!\!
f(A)g(B)h(C)\,.
\end{equation}

\begin{figure}[ht]
\begin{center}
\def\svgwidth{6cm} 
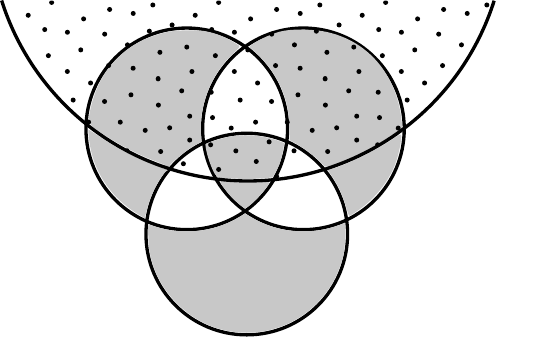 
\end{center} 
\vspace*{-4mm}
\caption{The set $(A\oplus B\oplus C)\cap Z$ (grey and dotted) from the definition of $T_p(Z)$.}
\label{fig:triple-intersect}
\end{figure}

The right-hand sides of the first system are now defined by
\begin{equation}
\label{eq:yi}
y_i=
\!\!\!\!\!\!\sum_{(u_1,u_2,\ldots,u_i)\in U^i}\!\!\!\!\!\!
T_0\bigl(\oplus_{\ell=1}^i\{u_\ell\}\bigr)
-T_1\bigl(\oplus_{\ell=1}^i\{u_\ell\}\bigr)\,.
\end{equation}
Let us recall that the universe $U$ has $n$ elements.
For nonnegative integers $i$ and $j$, let us define the
Vandermonde matrix $V=(v_{ij})$ by setting 
\begin{equation}
\label{eq:v}
v_{ij}=(n-2j)^i\,.
\end{equation}
{\em Remark.} We recall from basic linear algebra that 
any $d\times d$ submatrix of a $d\times e$ Vandermonde matrix 
with entries $z_j^i$ for $i=0,1,\ldots,d-1$ and $j=0,1,\ldots,e-1$ 
has nonzero determinant if the values $z_j$ are pairwise distinct. 
This makes a Vandermonde matrix particularly well-suited for 
building systems of independent equations from multiple families of equations.

\begin{Lem}[First family]
\label{lem:first}
For all $i=0,1,\ldots,d-1$ it holds that 
\begin{equation}
\label{eq:first-family}
\sum_{j\in J_q}v_{ij}x_j=y_i.
\end{equation}
\end{Lem}
\begin{proof}
Let us fix a triple $A,B,C\in\binom{U}{q}$ with 
$|A\oplus B\oplus C|=j$. From \eqref{eq:jq} we have $j\in J_q$.
Let us write $m_p(i)$ for the
number of $i$-tuples $(u_1,u_2,\ldots,u_i)\in U^i$ 
such that 
\begin{equation}
\label{eq:u-parity}
\bigl|\bigl(A\oplus B\oplus C\bigr)\cap\bigl(\oplus_{\ell=1}^i\{u_\ell\}\bigr)\bigr|\equiv p\pmod 2\,.
\end{equation}
From \eqref{eq:xj}, \eqref{eq:tp}, and \eqref{eq:yi} we observe 
that the lemma is implied by
\begin{equation}
\label{eq:vm}
v_{ij}=m_0(i)-m_1(i)\,.
\end{equation}
Indeed, $v_{ij}$ and $m_0(i)-m_1(i)$ are the coefficients before $f(A)g(B)h(C)$ in the LHS and RHS of~\eqref{eq:first-family}, respectively.
To prove~\eqref{eq:vm}, we proceed by induction on $i$. The base case $i=0$ is set up by
observing
\begin{equation}
\label{eq:m-base}
\begin{split}
m_0(0)&=1\,,\\
m_1(0)&=0\,.
\end{split}
\end{equation}
For $i\geq 1$, let us study what happens if we extend
an arbitrary $(i-1)$-tuple $(u_1,u_2,\ldots,u_{i-1})\in U^{i-1}$
by a new element $u_i\in U$. We observe that we have exactly $n-j$ 
choices for the value $u_i$ among the elements of $U$ outside 
$A\oplus B\oplus C$ and $j$ choices inside $A\oplus B\oplus C$. 
The parity \eqref{eq:u-parity} changes if and only if 
we choose an element inside $A\oplus B\oplus C$.
Thus, for $i\geq 1$ we have
\begin{equation}
\label{eq:m-step}
\begin{split}
   m_0(i)&=(n-j) m_0({i-1}) + jm_1(i-1)\,,\\
   m_1(i)&=j m_0({i-1}) + (n-j)m_1({i-1})\,.
\end{split}
\end{equation}
From \eqref{eq:m-base} and \eqref{eq:m-step} we thus have
\[
\begin{split}
  m_0(0)-m_1(0) &= 1\,,\\ 
  m_0(i)-m_1(i) &= (n-2j)\bigl(m_0({i-1})-m_1({i-1})\bigr)\,.
\end{split}
\]
Hence, $m_0(i)-m_1(i)=(n-2j)^i$ and from \eqref{eq:v} we conclude that the
lemma holds.

\end{proof}

\subsection{A second family of equations}

Our second family of equations is based on solving for the 
indeterminates \eqref{eq:xj} directly. We state the following
lemma in a general form, but for performance reasons we will 
in fact later use only the equations indexed by the $e-d$ smallest 
values $j\in J_q$ in our linear system.

\begin{Lem}[Second family]
\label{lem:second}
For all $j\in J_q$ it holds that
\begin{equation}
\label{eq:xj-direct}
x_j
=\!\!
\sum_{\substack{\ell=q-j%
}}^{q+j}
\sum_{D\in\binom{U}{\ell}}
\sum_{\substack{A,B\in\binom{U}{q}\\A\oplus B=D}}
f(A)g(B)
\!\!\!\!\!\!\!\!\!\!\!\!\sum_{\substack{C\in\binom{U}{q}\\|C\cap D|=(q+\ell-j)/2}}\!\!\!\!\!\!\!\!\!\!\!\!
h(C)\,.\!\!\!
\end{equation}
\end{Lem}
\begin{proof}
We must show that the right-hand side of \eqref{eq:xj-direct} equals
\eqref{eq:xj}. Let us study a triple $A,B,C\in\binom{U}{q}$ 
with $|A\oplus B\oplus C|=j$. We observe that $q-j\leq |A \oplus B|\leq q+j$ 
because otherwise taking the symmetric difference with $C$ will either 
leave too many elements uncanceled or it cannot cancel enough of 
the elements in $D=A\oplus B$. Since $|A|=|B|=q$ from \eqref{eq:xor-size} 
it follows that $|D|$ is in fact always even. 
Furthermore, when $|D|=\ell$ we observe that 
\[
\begin{split}
j&=|A\oplus B\oplus C|\\
&=|C\oplus D|\\
&=|C|+|D|-2|C\cap D|\\
&=q+\ell-2|C\cap D|\,.
\end{split}
\]
The lemma follows by solving for $|C\cap D|$ and observing that
each triple $A,B,C$ uniquely determines $D=A\oplus B$.
\end{proof}

\subsection{The linear system}
We are now ready to combine equations from the two families to a system 
\begin{equation}
\label{eq:sys}
M\vec x=\vec y 
\end{equation}
of independent linear equations for the indeterminates 
$\vec x=(x_j:j\in J_q)$. Recalling \eqref{eq:jq}, there are exactly
$e=\lfloor 3q/2\rfloor+1$ indeterminates, and hence exactly $e$
independent equations are required.

Let us use a parameter $0\leq\gamma\leq 1/2$ 
in building the system. (The precise value of $\gamma$ will be 
determined later in \S\ref{sect:runtime}.)

We now select $d=\lfloor (3/2-\gamma)q\rfloor+1$ equations 
from the first family (Lemma~\ref{lem:yi}), and $e-d$ equations
from the second family (Lemma~\ref{lem:xj-direct}).
More precisely, we access the first family for $d$ 
equations indexed by $i=0,1,\ldots,d-1$, and the second family 
for the $e-d$ equations indexed by the smallest $e-d$ values $j\in J_q$. 
That is, if $q$ is even, we use equations indexed by 
$j\in\{0,2,\ldots,2(e-d-1)\}$, and if $q$ is odd, we use 
equations indexed by $j\in\{1,3,\ldots,2(e-d)-1\}$.
Thus, for all $q$ we conclude that 
\[
i\leq\lfloor(3/2-\gamma)q\rfloor
\]
and that 
\[
\begin{split}
j
&\leq 2(e-d)-1\\
&= 2\bigl(\lfloor 3q/2\rfloor-\lfloor(3/2-\gamma)q\rfloor\bigr)-1\\
&\leq 2\bigl(\lfloor\gamma q\rfloor+1\bigr)-1\\
&\leq \lfloor 2\gamma q\rfloor+1\,.
\end{split}
\]

Let us now verify that the selected system consists of independent 
equations. To verify this it suffices to solve the system.
The equations from the second family (Lemma~\ref{lem:xj-direct}) 
by construction solve directly for $e-d$ indeterminates. 
We are thus left with $d$ equations from the first family
(Lemma~\ref{lem:yi}). Now observe that since we know
the values of $e-d$ indeterminates, we can subtract their
contribution from both sides of the remaining equations, leaving
us $d$ equations over $d$ indeterminates. In fact (see the remark 
before Lemma~\ref{lem:yi}), the coefficient matrix of the remaining
system is a $d\times d$ submatrix of the original Vandermonde matrix,
and hence invertible. We conclude that the equations are independent.

It remains to argue that the system \eqref{eq:sys} can be constructed
and solved within the claimed running time.

\section{Efficient construction and solution}
\label{sect:solving}

This section proves Theorem~\ref{thm:main} by constructing and
solving the system derived in \S\ref{sect:system} within
the claimed running time. We start with some useful subroutines 
that enable us to efficiently construct the right-hand sides 
for \eqref{eq:yi} and \eqref{eq:xj-direct}.

\subsection{The intersection transform}

Let $s$ and $0\leq t\leq s$ be nonnegative integers.
For a function $f:\binom{U}{q}\rightarrow R$, define 
the {\em intersection transform} $f\iota_t:\binom{U}{s}\rightarrow R$ of $f$
for all $Z\in\binom{U}{s}$ by
\begin{equation}
\label{eq:ft}
f\iota_t(Z)=\!\sum_{\substack{A\in\binom{U}{q}\\|A\cap Z|=t}}\!f(A)\,.
\end{equation}

The following lemma is an immediate corollary of a theorem of Bj\"orklund~{\em et~al.}~\cite[Theorem 1]{BHKK08}.

\begin{Lem}
\label{lem:fit}
There exists an algorithm that evaluates all the ${|U| \choose s}$ values of the intersection transform
for all $0\leq t\leq s$ in time $O\bigl(n^{\max(s,q)+c}\bigr)$ for a 
constant $c\geq 0$ independent of constants $s$ and $q$. 
\end{Lem}

\noindent
{\em Remark.} 
Lemma~\ref{lem:fit} can be stated in an alternative form that counts 
the number of arithmetic operations 
(addition, subtraction, multiplication, and exact division of integers) 
performed by the algorithm on the input $f$ to obtain $f\iota_t$ for all 
$0\leq t\leq s$. This form is obtained by simply removing the constant 
$c$ from the bound in Lemma~\ref{lem:fit}.
(Indeed, we can use Bareiss's algorithm~\cite{B68} to solve the 
underlying linear system with exact divisions.)

\subsection{The parity transform}
\label{sect:parity-trans}

Let $s$ be a nonnegative integer and let $p\in\{0,1\}$. 
For a function $f:\binom{U}{q}\rightarrow R$, define 
the {\em parity transform} $f\pi_p:\binom{U}{s}\rightarrow R$ of $f$
for all $Z\in\binom{U}{s}$ by
\begin{equation}
\label{eq:fp}
f\pi_p(Z)=
\!\!\!\!\!\!\!\!\!\sum_{\substack{A\in\binom{U}{q}\\|A\cap Z|\equiv p\!\!\!\pmod 2}}\!\!\!\!\!\!\!\!\!
f(A)\,.
\end{equation}

\begin{Lem}
\label{lem:fpt}
There exists an algorithm that evaluates the parity transform
for $p\in\{0,1\}$ in time $O\bigl(n^{\max(s,q)+c}\bigr)$ for a 
constant $c\geq 0$ independent of constants $s$ and $q$. 
\end{Lem}
\begin{proof}
We observe that 
\[
f\pi_p=\sum_{\substack{t\in\{0,1,\ldots,s\}\\t\equiv p\!\!\!\pmod 2}}f\iota_t
\]
and apply Lemma~\ref{lem:fit}.
\end{proof}

\subsection{Evaluating the right-hand side of the first family}
Let $i$ be a nonnegative integer. Our objective is to evaluate 
the right-hand side of \eqref{eq:yi}.
Let us start by observing that it suffices to compute the values
\eqref{eq:tp} for all $Z\subseteq U$ with $|Z|\leq i$. 

The following lemma will be useful towards this end.
Denote by $L_n(i,s)$ the number of tuples $(u_1,u_2,\ldots,u_i)\in U^i$ with 
$s=|\oplus_{\ell=1}^i\{u_\ell\}|$ and $n=|U|$.

\begin{Lem}
\label{lem:lnis}
We have
\[
L_n(i,s)=\begin{cases}
1 & \text{if $i=s=0$};\\
0 & \text{if $i<s$};\\
L_n(i-1,1) & \text{if $i\geq 1$ and $s=0$};\\[2.5mm]
(n-s+1)L_n(i-1,s-1)& \\
\quad\ +(s+1)L_n(i-1,s+1) & \text{if $i\geq s\geq 1$}.\\
\end{cases}
\]
In particular, the values $L_n(i,s)$ can be computed for all
$0\leq s\leq i\leq 3q/2\leq n$ in time $O(q^2)$.
\end{Lem}
\begin{proof}
When we insert an element $u_i$ into a tuple we may obtain a tuple with exactly $s\geq 1$ elements that occur an odd number of times in two different ways: 
either we had $s-1$ such elements and insert a new element ($n-s+1$ choices), 
or we had $s+1$ such elements and insert one of them. 
The running time follows by tabulating the values $L_n(i,s)$ in increasing 
lexicographic order in $(i,s)$. This completes the lemma.
\end{proof}

Now let us reduce \eqref{eq:yi} to \eqref{eq:tp}. In particular,
we have 
\begin{equation}
\label{eq:yi-tp}
y_i=\sum_{s=0}^i \tbinom{n}{s}^{-1}L_n(i,s)\sum_{Z\in\binom{U}{s}}T_0(Z)-T_1(Z)\,.
\end{equation}
Indeed, by symmetry each $Z\in{U\choose s}$ has the same number $l_{i,s}$ of tuples $(u_1,u_2,\ldots,u_i)$ such that $|\!\oplus_{\ell=1}^i\{u_{\ell}\}|=Z$. Hence, $L_n(i,s)={n \choose s}l_{i,s}$ and~\eqref{eq:yi-tp} follows.
Thus it remains to compute the values \eqref{eq:tp}.
At this point it is convenient to recall Fig.~\ref{fig:triple-intersect}.
We have
\begin{equation}
\label{eq:xor-cap}
\begin{split} 
\bigl|(A\oplus B\oplus C)\cap Z\bigr|
&= 
|(A \cap Z) \oplus(B \cap Z) \oplus(C \cap Z)| \\
&\equiv
|A\cap Z|
+|B\cap Z|
+|C\cap Z|\pmod 2\,,
\end{split}
\end{equation}
where the congruence follows from~\eqref{eq:xor-size-equiv-sum}.
Let us use the shorthand $\bar p=1-p$ for the complement of $p\in\{0,1\}$.
Denoting pointwise multiplication of functions by ``$\cdot$'',
from \eqref{eq:xor-cap}, \eqref{eq:tp} and \eqref{eq:fp} it immediately
follows that
\begin{equation}
\label{eq:tpfgh}
\begin{split}
T_p
=f\pi_p\cdot g\pi_p\cdot h\pi_p
+f\pi_{\bar p}\cdot g\pi_{\bar p}\cdot h\pi_p
+f\pi_{\bar p}\cdot g\pi_p\cdot h\pi_{\bar p}
+f\pi_p\cdot g\pi_{\bar p}\cdot h\pi_{\bar p}\,.
\end{split}
\end{equation}

\begin{Lem}
\label{lem:yi}
There exists an algorithm that for $0\leq\gamma\leq 1/2$ evaluates
the right-hand side \eqref{eq:yi} for all $0\leq i\leq(3/2-\gamma)q$
in time $O(n^{(3/2-\gamma)q+c})$ for a constant $c\geq 0$ 
independent of constants $\gamma$ and $q$.
\end{Lem}
\begin{proof}
First evaluate the parity transforms $f\pi_p$, $g\pi_p$, $h\pi_p$ 
for all $p\in\{0,1\}$ and $s\leq (3/2-\gamma)q$ using Lemma~\ref{lem:fpt}.
Then use \eqref{eq:tpfgh} to evaluate $T_p$ 
for all $p\in\{0,1\}$ and $s\leq (3/2-\gamma)q$. Finally
compute the coefficients $L_n(i,s)$ using Lemma~\ref{lem:lnis} 
and evaluate the right-hand sides \eqref{eq:yi} via \eqref{eq:yi-tp}.
\end{proof}

\subsection{The symmetric difference product}
\label{sect:sdp}

Let $\ell$ be a nonnegative even integer. For $f,g:\binom{U}{q}\rightarrow R$
define the symmetric difference product 
$f\oplus g:\binom{U}{\ell}\rightarrow R$ for all $D\in\binom{U}{\ell}$
by
\begin{equation}
\label{eq:f-oplus-g}
(f\oplus g)(D)=
\!\!\sum_{\substack{A,B\in\binom{U}{q}\\A\oplus B=D}}\!\!f(A)g(B)\,.
\end{equation}
From \eqref{eq:xor-size} we observe that if $|A\oplus B|=\ell$ 
with $|A|=|B|=q$ then $|A\cap B|=q-\ell/2$ and 
$|A\setminus B|=|B\setminus A|=\ell/2$. Define the matrix
$F$ with rows indexed by $I\in\binom{U}{\ell/2}$ and
columns indexed by $K\in\binom{U}{q-\ell/2}$ with
the $(I,K)$-entry defined by
\begin{equation}
\label{eq:F}
F(I,K)=
\begin{cases}
f(I\cup K) & \text{if $I\cap K=\emptyset$};\\
0          & \text{otherwise}.
\end{cases}
\end{equation}
Define the matrix
$G$ with rows indexed by $K\in\binom{U}{q-\ell/2}$
and columns indexed by $J\in\binom{U}{\ell/2}$ with
the $(K,J)$-entry defined by
\begin{equation}
\label{eq:G}
G(K,J)=
\begin{cases}
g(K\cup J) & \text{if $K\cap J=\emptyset$};\\
0          & \text{otherwise}.
\end{cases}
\end{equation}
From \eqref{eq:F} and \eqref{eq:G}
we observe that the product matrix $FG$ enables us to recover
the symmetric difference product \eqref{eq:f-oplus-g} for all 
$D\in\binom{U}{\ell}$ by 
\begin{equation}
\label{eq:from-fg}
(f\oplus g)(D)=\sum_{I\in\binom{D}{\ell/2}}FG\bigl(I,D\setminus I\bigr)\,.
\end{equation}

Recall that we write $\omega$ for the limiting exponent of square matrix 
multiplication, $2\leq \omega<2.3728639$, and $\alpha$ for the limiting exponent such that multiplying an $N\times N^\alpha$ matrix with an $N^\alpha\times N$ matrix takes $N^{2+o(1)}$ arithmetic operations, $0.30<\alpha\leq 3-\omega$.
For generic rectangular matrices it is known (cf.~\cite{LR83}) that the product of 
an $N\times M$ matrix and an $M\times N$ matrix can be computed using

\medskip
\noindent
\begin{tabular}{@{}l@{\ \ }p{11.5cm}@{}}
(M1) & $O(N^{2-\alpha \beta}M^\beta+N^2)$ arithmetic operations, with $\beta=(\omega-2)/(1-\alpha)$ when $M\leq N$, and\\[\medskipamount]
(M2) & $O(N^{\omega-1}M)$ arithmetic operations 
by decomposing the product into a sum of $\lceil M/N\rceil$ products 
of $N\times N$ square matrices when $M\geq N$.
\end{tabular}
\medskip

\noindent
{\em Remark.} The bounds above are not the best possible \cite{LG12}; 
however, to provide a clean exposition we will work with these 
somewhat sub-state-of-the-art bounds.

\begin{Lem}
\label{lem:fsdp}
There exists an algorithm that for $0\leq\gamma\leq 1/2$
evaluates the symmetric difference product \eqref{eq:f-oplus-g} 
for all even $(1-2\gamma)q-1\leq\ell\leq(1+2\gamma)q+1$ in time 
\[
O\bigl(n^{\omega q/2+(2-\alpha\beta-\beta)\gamma q+c}+n^{(1+2\gamma)q+c}\bigr)
\] 
for a constant $c\geq 0$ independent of constants $\gamma$ and $q$. 
\end{Lem}
\begin{proof}
For convenience we may pad $F$ and $G$ with all-zero rows and columns 
so that $F$ becomes an $n^{\ell/2}\times n^{q-\ell/2}$ matrix
and $G$ becomes an $n^{q-\ell/2}\times n^{\ell/2}$ matrix.

For $\ell\leq q$ by (M2) we can thus multiply $F$ and $G$ using
$O(n^{\omega \ell/2+q-\ell})\subseteq O(n^{\omega q/2})$ arithmetic 
operations and hence in time $O(n^{\omega q/2+c})$. 

For $\ell\geq q$ by (M1) we can thus multiply $F$ and $G$ using 
\[
O\bigl(n^{(2-\alpha\beta)\ell/2+\beta(q-\ell/2)}+n^\ell\bigr)
\]
arithmetic operations.
For $q\leq\ell\leq(1+2\gamma)q+1$ the linear function 
\[
u(\ell)=(2-\alpha\beta)\ell/2+\beta(q-\ell/2)
\]
has its maximum at $\ell = q$ or at $\ell=(1+2\gamma)q+1$. 
Noting that $2-\alpha\beta+\beta=\omega\geq 2$, at $\ell=q$ we
obtain the bound $O(n^{\omega q/2+c})$ for the running time.
At $\ell=(1+2\gamma)q+1$ we obtain the running time bound 
\[
O\bigl(n^{(2-\alpha\beta)(1/2+\gamma)q+\beta(1/2-\gamma)q+c}+n^{(1+2\gamma)q+c}\bigr)\,.
\] 
Again noting that $2-\alpha\beta+\beta=\omega$, this bound simplifies to 
\[
O\bigl(n^{\omega q/2+(2-\alpha\beta-\beta)\gamma q+c}+n^{(1+2\gamma)q+c}\bigr)\,.
\] 
It remains to analyze the quantity $2-\alpha\beta-\beta$. 
Towards this end, recall that $0.3<\alpha\leq 3-\omega$ and 
$2\leq\omega<2.3728639$, with $\alpha=1$ if and only 
if $\omega=2$. The lemma now follows by observing that 
\[
\begin{split}
2-\alpha\beta-\beta 
&=\omega-2\beta\\
&=\omega-\frac{2(\omega-2)}{1-\alpha}\\
&= \frac{4-(1+\alpha)\omega}{1-\alpha}\\
&\geq \frac{4-(4-\omega)\omega}{1-\alpha}\\
&\geq 0\,.
\end{split}
\]
\end{proof}

\subsection{Evaluating the right-hand side of the second family}

Let $j$ be a nonnegative integer. Our objective is to evaluate 
the right-hand side of \eqref{eq:xj-direct}. Let us start
by observing that \eqref{eq:xj-direct} can be stated using
the symmetric difference product \eqref{eq:f-oplus-g} 
and the intersection transform \eqref{eq:ft} in equivalent
form
\begin{equation}
\label{eq:xj-direct-2}
x_j=
\sum_{\substack{\ell=q-j%
}}^{q+j}
\sum_{D\in\binom{U}{\ell}}
(f\oplus g)(D)\cdot h\iota_{(q+\ell-j)/2}(D)\,.
\end{equation}

\begin{Lem}
\label{lem:xj-direct}
There exists an algorithm that for $0\leq\gamma\leq 1/2$
evaluates the right-hand-side of \eqref{eq:xj-direct}
for all $0\leq j\leq 2\gamma q+1$ in time 
\[
O\bigl(n^{\omega q/2+(2-\alpha\beta-\beta)\gamma q+c}+n^{(1+2\gamma)q+c}\bigr)
\] 
for a constant $c\geq 0$ independent of constants $\gamma$ and $q$. 
\end{Lem}
\begin{proof}
Because $0\leq j\leq 2\gamma q+1$ we observe that
$(1-2\gamma)q-1\leq\ell\leq(1+2\gamma)q+1$ in \eqref{eq:xj-direct-2}. 
Using Lemma~\ref{lem:fsdp} we can evaluate $f\oplus g$ for
all required $\ell$ within the claimed time bound. 
Using Lemma~\ref{lem:fit} with 
$s=\lfloor(1+2\gamma)q+1\rfloor$ we can evaluate $h\iota_t$ for all 
$t\leq s$ in $O(n^{(1+2\gamma)q+c})$ time which is within the claimed time bound. Finally,
the sums in \eqref{eq:xj-direct-2} can be computed
in the claimed time bound by using the evaluations $f\oplus g$
and $h\iota_t$.
\end{proof}

\subsection{Running-time analysis}
\label{sect:runtime}

We now balance the running times from Lemma~\ref{lem:yi} 
and Lemma~\ref{lem:xj-direct} by selecting the value of 
$0\leq\gamma\leq 1/2$. Disregarding the constant $c$ which
is independent of $q$ and $\gamma$, the contribution
of Lemma~\ref{lem:yi} is $O\bigl(n^{(3/2-\gamma)q}\bigr)$ 
and the contribution of Lemma~\ref{lem:xj-direct} is 
\[
O\bigl(n^{\omega q/2+(2-\alpha\beta-\beta)\gamma q}+n^{(1+2\gamma)q}\bigr)\,.
\]
In particular, we must minimize the maximum of the three contributions
\[
\begin{split}
&O\bigl(n^{(3/2-\gamma)q}\bigr)\,,\\
&O\bigl(n^{\omega q/2+(2-\alpha\beta-\beta)\gamma q}\bigr)\,,\quad\text{and}\\
&O\bigl(n^{(1+2\gamma)q}\bigr)\,.
\end{split}
\]
We claim that if $\alpha\leq 1/2$ then the maximum is controlled by
\begin{equation}
\label{eq:gamma-bal}
\frac{3}{2}-\gamma=\frac{\omega}{2}+(2-\alpha\beta-\beta)\gamma\,.
\end{equation}
Let us select the value of $\gamma$ given by \eqref{eq:gamma-bal}.
Recalling that $\beta=(\omega-2)/(1-\alpha)$, we have
\begin{equation}
\label{eq:gamma}
\gamma
=\frac{3-\omega}{2(3-\alpha\beta-\beta)}
=\frac{(3-\omega)(1-\alpha)}{12-2(1+\omega)(1+\alpha)}\,.
\end{equation}
In \eqref{eq:gamma} we have $\gamma=1/6$ if and only if $\alpha=1/2$.
In particular, we have $\gamma\leq 1/6$ if $\alpha\leq 1/2$, implying
\[
\frac{3}{2}-\gamma\geq 1-2\gamma
\]
and thus \eqref{eq:gamma} and \eqref{eq:gamma-bal} determine the maximum
as claimed. In this case we can achieve running time
\[
\begin{split}
O\biggl(&n^{\bigl(\frac{3}{2}-\frac{(3-\omega)(1-\alpha)}{12-2(1+\omega)(1+\alpha)}\bigr)q+c}\biggr)\\
&=O\biggl(n^{3q\bigl(\frac{1}{2}-\frac{(3-\omega)(1-\alpha)}{36-6(1+\omega)(1+\alpha)}\bigr)+c}\biggl)\,.
\end{split}
\]
Conversely, if $\alpha\geq 1/2$ then the maximum is controlled by
\[
\frac{3}{2}-\gamma=1-2\gamma\,,
\]
in which case we select $\gamma=1/6$ and achieve running time
\[
O\biggr(n^{\bigl(\frac{3}{2}-\frac{1}{6}\bigr)q+c}\biggr)
=O\biggl(n^{3q\bigl(\frac{1}{2}-\frac{1}{18}\bigr)+c}\biggr)\,.
\]

Since the system \eqref{eq:sys} and its solution \eqref{eq:xj} are 
integer-valued and have bit-length bounded by a polynomial in $n$
that is independent of the constant $q$, for example Bareiss's
algorithm~\cite{B68} solves the constructed system in 
the claimed running time.
This completes the proof of Theorem~\ref{thm:main}.

\subsection{Speedup for $q=2,3,4$.}
\label{sect:q234}

In this section we prove Theorem~\ref{thm:q234}. We split the proof
into three parts.

\begin{proof}[Proof ($q=2$).]
Let us study the first family of equations (Lemma~\ref{lem:first}). 
For $q=2$ we have indeterminates $x_0,x_2,x_4,x_6$ and
equations indexed by $i=0,1,2,3$, where equation $i$ can be 
constructed in $O(n^{\max(i,q)})$ arithmetic operations;
cf.~Lemma~\ref{lem:fit} to Lemma~\ref{lem:yi}.
Thus, it suffices to replace the equation for $i=3$ with an 
equation independent of the equations $i=0,1,2$ to solve
for all the indeterminates, and in particular for $x_6$, which
gives the weighted disjoint triples. Our strategy is to solve 
directly for the indeterminate $x_0$. We observe that $x_0$ requires
to sum over all triples $(A,B,C)$ of $q$-subsets such that
the $q$-uniform hypergraph $\{A,B,C\}$ has no vertices of 
odd degree. 
Up to isomorphism the only such hypergraph
for $q=2$ is the triangle.
From now on we abuse the notation slightly and extend the domains of the functions $f$, $g$ and $h$ to sets of size {\em at most $q$} so that they evaluate to $0$ for sets of size strictly smaller than $q$.
Accordingly, we have
\begin{equation}
\label{eq:2-uniform-x0} 
x_0=\sum_{1\leq p,r\leq n}f\bigl(\{p,r\}\bigr)\sum_{s=1}^ng\bigl(\{p,s\}\bigr)h\bigl(\{r,s\}\bigr)\,,
\end{equation}

where we can evaluate the inner sum simultaneously for all $p,r$ 
by multiplying two $n\times n$ matrices using $O(n^{\omega})$
arithmetic operations.
\end{proof}

\begin{proof}[Proof ($q=3$).]
Let us imitate the proof for $q=2$.
For $q=3$ our indeterminates are $x_1,x_3,x_5,x_7,x_9$, and the
equations are indexed by $i=0,1,2,3,4$. Again it suffices to replace
the $i=4$ equation. We will do this by solving directly for
the indeterminate $x_1$. For $q=3$ there are, up to isomorphism, 
exactly two $q$-uniform hypergraphs $\{A,B,C\}$ with a unique vertex $p$
of odd degree. 
In Type I hypergraphs $p$ is of degree $3$ and in Type II hypergraphs $p$ is of degree $1$. 
Let $x_{1,\mathrm{I}}$ and $x_{1,\mathrm{II}}$ denote the contribution to $x_1$ of triples $(A,B,C)$ corresponding to Type I and Type II hypergraphs, respectively.
Then $x_1=x_{1,\mathrm{I}}+x_{1,\mathrm{II}}$.

Note that for every Type I hypergraph the hypergraph $\{A\setminus\{p\},B\setminus\{p\},C\setminus\{p\}\}$ is 2-uniform and has no odd vertices.
Hence the contribution $x_{1,\mathrm{I},p}$ of Type I triples such that $p\in A\cap B\cap C$ can be computed in time $O(n^{\omega})$ by applying the formula~\eqref{eq:2-uniform-x0} to functions $f_p,g_p,h_p:{U\choose 2}\rightarrow \mathbb{Z}$, where $f_p(X)=f(\{p\}\cup X)$, $g_p(X)=g(\{p\}\cup X)$ and $h_p(X)=h(\{p\}\cup X)$. (Note that we use here the fact that $f$, $g$ and $h$ evaluate to $0$ for sets with less than 3 elements.)
Since \[x_{1,\mathrm{I}} = \sum_{p=1}^n x_{1,\mathrm{I},p},\]
the value $x_{1,\mathrm{I}}$ can be computed in $O(n^{\omega+1})$ time.

\begin{figure}[t]
\[
\mathrm{I}:\ 
\begin{array}{|@{\,}c@{\,}|@{\,}c@{\,}|}
\hline
p&rst\\
\hline
1&110\\
1&101\\
1&011\\
\hline
\end{array}\,,
\quad
\mathrm{II}:\ 
\begin{array}{|@{\,}c@{\,}|@{\,}c@{\,}|@{\,}c@{\,}|}
\hline
p&rs&tu\\
\hline
1&11&00\\\hline
0&10&11\\
0&01&11\\
\hline
\end{array}\,.
\]
\caption{\label{fig:3-uniform}Two nonisomorphic types of 3-uniform hypergraphs with one vertex of odd degree.
We display these hypergraphs below as incidence matrices
where the rows correspond to hyperedges and the columns correspond
to vertices, with a 1-entry indicating indidence and a 0-entry 
indicating non-indidence between a hyperedge and a vertex.
Vertical and horizontal lines to partition the vertices and 
the hyperedges to orbits with respect to the action of the automorphism 
group of the hypergraph.}
\end{figure}

Now consider Type II hypergraphs. 
Let us further partition Type II triples $(A,B,C)$ according to which of the sets $A$, $B$, $C$ contains $p$.
Let $z_{1,\mathrm{II}}^{f|g,h}$ denote the contribution to $x_1$ of Type II triples where $p\in A$.
Note that then $z_{1,\mathrm{II}}^{g|f,h}$ and $z_{1,\mathrm{II}}^{h|f,g}$ are the contributions of Type II triples where $p\in B$ and $p\in C$, respectively, and hence 
\begin{equation}
\label{eq:3-uniform-type-II}
 x_{1,\mathrm{II}}=z_{1,\mathrm{II}}^{f|g,h}+z_{1,\mathrm{II}}^{g|f,h}+z_{1,\mathrm{II}}^{h|f,g}.
\end{equation}

Let us focus on $z_{1,\mathrm{II}}^{f|g,h}$, i.e.\ assume that $A=\{p,r,s\}$ for some $r,s\in U$.
Since $r$ and $s$ are of degree $2$, either both are in one of the remaining sets, say $B$, or each of $r$, $s$ is in exactly one of $B$ and $C$.
However we can assume the latter, because in the former $C$ has at least two degree $1$ vertices.
So let $r$ be the vertex of $A\cap B$ and let $s$ be the vertex of $A\cap C$.
Since the remaining vertices in $B\cup C$ are of degree 2, there are exactly two of them, say, $t$ and $u$, and $\{t,u\}\in B\cap C$, see Fig~\ref{fig:3-uniform}. It follows that 

\[
\begin{split}
z_{1,\mathrm{II}}^{f|g,h}
=& 
\sum_{1\leq p,r,s\leq n}
\sum_{\substack{1\leq t<u\leq n\\p\not\in\{t,u\}}}
f\bigl(\{p,r,s\}\bigr)
g\bigl(\{r,t,u\}\bigr)
h\bigl(\{s,t,u\}\bigr)\\
=&
\sum_{1\leq p,r,s\leq n}
f\bigl(\{p,r,s\}\bigr)
\biggl(\sum_{1\leq t<u\leq n}
g\bigl(\{r,t,u\}\bigr)
h\bigl(\{s,t,u\}\bigr)
\biggr) \ -\\
& 
\underbrace{
\sum_{1\leq p,r,s\leq n}
\sum_{1\leq t\leq n}
f\bigl(\{p,r,s\}\bigr)
g\bigl(\{p,r,t\}\bigr)
h\bigl(\{p,s,t\}\bigr)}_{x_{1,\mathrm{I}}}\,.
\end{split}
\]
Note that it is sufficient to assume only $p\not\in\{t,u\}$; indeed, since $f$, $g$ and $h$ evaluate to $0$ for sets with less than 3 elements any choice of $p$, $r$, $s$, $t$, $u$ which satisfies this assumption but $|\{p,r,s,t,u\}|<5$ produces a zero term in the sum.
Here the sum $\sum_{1\leq p,r,s\leq n}
f\bigl(\{p,r,s\}\bigr)
\biggl(\sum_{1\leq t<u\leq n}
g\bigl(\{r,t,u\}\bigr)
h\bigl(\{s,t,u\}\bigr)
\biggr)$ can be evaluated with an $n\times n^2$ by 
$n^2\times n$ rectangular matrix multiplication in $O(n^{1+\omega})$
arithmetic operations; cf.~(M2). 
Hence it takes $O(n^{1+\omega})$ time to compute $z_{1,\mathrm{II}}^{f|g,h}$, since we have shown that $x_{1,\mathrm{I}}$ can also be computed within this time bound.
\end{proof}

\begin{proof}[Proof ($q=4$).]
Let us imitate the proof for $q=3$.
For $q=4$ our indeterminates are $x_0,x_2,x_4,x_6,x_8,x_{10},x_{12}$, 
and the equations are indexed by $i=0,1,2,3,4,5,6$. 
It suffices to replace the $i=5$ and $i=6$ equations. 
We will do this by solving directly for
the indeterminates $x_0$ and $x_2$, that is, the cases $j=0$ and $j=2$
for $j=|A\oplus B\oplus C|$.

{\em The case $j=0$.}
For $q=4$ there is, up to isomorphism, a unique 
$q$-uniform hypergraph $\{A,B,C\}$ with no vertex of odd degree:
\[
\begin{array}{|@{\,}c@{\,}|}
\hline
111100\\
110011\\
001111\\
\hline
\end{array}\,.
\]
Accordingly, we have
\[
x_0=\sum_{1\leq p<q\leq n}\sum_{1\leq r<s\leq n}f\bigl(\{p,q,r,s\}\bigr)\sum_{1\leq t<u\leq n}g\bigl(\{p,q,t,u\}\bigr)h\bigl(\{r,s,t,u\}\bigr)\,,
\]
where we can evaluate the inner sum simultaneously for all $p,q,r,s$ 
by multiplying two $n^2\times n^2$ matrices. This takes $O(n^{2\omega})$
arithmetic operations.

{\em The case $j=2$.}
For $q=4$ we will show that there are, up to isomorphism, exactly four 
$q$-uniform hypergraphs $\{A,B,C\}$ with exactly two
vertices $p,r$ of odd degree (see Fig~\ref{fig:4-uniform-j=2}).
For $t\in \{\mathrm{I},\mathrm{II},\mathrm{III},\mathrm{IV}\}$ let $x_{2,t}$ denote the contribution to $x_2$ of triples $(A,B,C)$ such that the corresponding hypergraph is of type $t$.

\begin{figure}[t]
\[
\mathrm{I}:\ 
\begin{array}{|@{\,}c@{\,}|@{\,}c@{\,}|}
\hline
pr&stu\\\hline
11&110\\
11&101\\
11&011\\
\hline
\end{array}\,,
\quad
\mathrm{II}:\ 
\begin{array}{|@{\,}c@{\,}|@{\,}c@{\,}|@{\,}c@{\,}|@{\,}c@{\,}|}
\hline
p&r&st&uv\\\hline
1&1&11&00\\\hline
1&0&10&11\\
1&0&01&11\\
\hline
\end{array}\,,
\quad
\mathrm{III}:\ 
\begin{array}{|@{\,}c@{\,}|@{\,}c@{\,}|@{\,}c@{\,}|}
\hline
pr&st&uvw\\\hline
11&11&000\\\hline
00&10&111\\
00&01&111\\
\hline
\end{array}\,,
\quad
\mathrm{IV}:\ 
\begin{array}{|@{\,}c@{\,}|@{\,}c@{\,}|@{\,}c@{\,}|}
\hline
pr&stuv&w\\\hline
10&1100&1\\
01&0011&1\\\hline
00&1111&0\\
\hline
\end{array}\,.
\]
\caption{\label{fig:4-uniform-j=2}Four nonisomorphic $4$-uniform hypergraphs with exactly two vertices of odd degree.}
\end{figure}

In hypergraphs of Type I both odd degree vertices $p$ and $r$ are of degree 3.
Then $\{A\setminus\{p,r\},B\setminus\{p,r\},C\setminus\{p,r\}\}$ is 2-uniform and has no odd vertices.
Hence the contribution $x_{2,\mathrm{I},p,r}$ of Type I triples such that both $p$ and $r$ are of degree $3$ can be computed in time $O(n^{\omega})$ by applying the formula~\eqref{eq:2-uniform-x0} to functions $f_{pr},g_{pr},h_{pr}:{U\choose 2}\rightarrow \mathbb{Z}$, where $f_{pr}(X)=f(\{p,r\}\cup X)$, $g_{pr}(X)=g(\{p,r\}\cup X)$ and $h_{pr}(X)=h(\{p,r\}\cup X)$. 
Since \[x_{2,\mathrm{I}} = \sum_{1\le p<r\le n} x_{2,\mathrm{I},p,r},\]
the value $x_{2,\mathrm{I}}$ can be computed in $O(n^{\omega+2})$ time.

In hypergraphs of Type II there is one vertex $p$ of degree $3$ and one vertex $r$ of degree $1$.
Then $\{A\setminus\{p\},B\setminus\{p\},C\setminus\{p\}\}$ is 3-uniform and has exactly one odd degree vertex, in fact a degree $1$ vertex.
Hence the contribution $x_{2,\mathrm{II},p}$ of Type II triples such that $p$ is of degree $3$ can be computed in time $O(n^{\omega+1})$ by applying the formula~\eqref{eq:3-uniform-type-II} to functions $f_{p},g_{p},h_{p}:{U\choose 3}\rightarrow \mathbb{Z}$, where $f_{p}(X)=f(\{p\}\cup X)$, $g_p(X)=g(\{p\}\cup X)$ and $h_p(X)=h(\{p\}\cup X)$. 
Since \[x_{2,\mathrm{II}} = \sum_{1\le p\le n} x_{2,\mathrm{II},p},\]
the value $x_{2,\mathrm{II}}$ can be computed in $O(n^{\omega+2})$ time.
Note also that by the same reasoning the contribution $z_{2,\mathrm{II}}^{f|g,h}$ of Type II triples $(A,B,C)$ such that the degree 1 vertex belongs to $A$ is equal to $\sum_{1\le p\le n} z_{1,\mathrm{II}}^{f_p|g_p,h_p}$, hence it also can be computed in $O(n^{\omega+2})$ time. (We will use this quantity while computing $x_{2,\mathrm{III}}$ and $x_{2,\mathrm{IV}}$)

In the remaining hypergraphs both $p$ and $r$ are of degree $1$, but we have two nonisomorphic types of such hypergraphs.
In Type III hypergraphs both $p$ and $r$ are in the same hyperedge.
Let $z_{2,\mathrm{III}}^{f|g,h}$ denote the contribution to $x_2$ of Type III triples where $p,r\in A$.
Note that then $z_{2,\mathrm{III}}^{g|f,h}$ and $z_{2,\mathrm{III}}^{h|f,g}$ are the contributions of Type III triples where $p,r\in B$ and $p,r\in C$, respectively, and hence 
\begin{equation}
\label{eq:4-uniform-type-III}
 x_{2,\mathrm{III}}=z_{2,\mathrm{III}}^{f|g,h}+z_{2,\mathrm{III}}^{g|f,h}+z_{2,\mathrm{III}}^{h|f,g}.
\end{equation}
We focus on $z_{2,\mathrm{III}}^{f|g,h}$, i.e.\ we assume $A=\{p,r,s,t\}$ for some vertices $s$ and $t$ such that $|A|=4$.
Then since the remaining vertices are all of degree $2$, the remaining sets in the triple are $B=\{s,u,v,w\}$ and $C=\{t,u,v,w\}$ for some three different vertices $u$, $v$, $w$ outside $A$. Then,

\[
\begin{split}
z_{2,\mathrm{III}}^{f|g,h}
=& 
\sum_{1\leq p<r\leq n}
\sum_{1\leq s,t\leq n}
\sum_{\substack{1\leq u<v<w\leq n\\\{p,r\}\cap\{u,v,w\}=\emptyset}}
f\bigl(\{p,r,s,t\}\bigr)
g\bigl(\{s,u,v,w\}\bigr)
h\bigl(\{t,u,v,w\}\bigr)\\
=&
\underbrace{\sum_{1\leq p<r\leq n}
\sum_{1\leq s,t\leq n}
f\bigl(\{p,r,s,t\}\bigr)
\biggl(\sum_{1\leq u<v<w\leq n}
g\bigl(\{s,u,v,w\}\bigr)
h\bigl(\{t,u,v,w\}\bigr)
\biggr)}_{(*)} \ -\\
& 
\underbrace{
\sum_{1\leq p<r\leq n}
\sum_{1\leq s,t\leq n}
\sum_{\substack{1\leq u<v<w\leq n\\\{p,r\}\cap\{u,v,w\}\ne\emptyset}}
f\bigl(\{p,r,s,t\}\bigr)
g\bigl(\{s,u,v,w\}\bigr)
h\bigl(\{t,u,v,w\}\bigr)}_{o^{f|g,h}}\,,
\end{split}
\]
where $o^{f|g,h}$ is an overcount which we specify in a moment.
Note that it is sufficient to assume only $\{p,r\}\cap\{u,v,w\}=\emptyset$; indeed, since $f$, $g$ and $h$ evaluate to $0$ for sets with less than 4 elements any choice of $p$, $r$, $s$, $t$, $u$, $v$, $w$ which satisfies this assumption but $|\{p,r,s,t,u,v,w\}|<7$ produces a zero term in the sum.
Observe also that the sum $\sum_{1\leq u<v<w\leq n}
g\bigl(\{s,u,v,w\}\bigr)
h\bigl(\{t,u,v,w\}\bigr)$ can be evaluated for each $s,t$ 
with an $n\times n^3$ by $n^3\times n$ rectangular matrix multiplication
in $O(n^{2+\omega})$ arithmetic operations; cf.~(M2).
Once these values are tabulated for every $s,t$, the sum $(*)$ can be evaluated in $O(n^4)\subseteq O(n^{2+\omega})$ time.
It remains to compute the value of overcount $o^{f|g,h}$ efficiently.
This can be split as $o^{f|g,h}=o^{f|g,h}_1+o^{f|g,h}_2$, where $o^{f|g,h}_1$ and $o^{f|g,h}_2$ denote the contribuitions of the terms where $|\{p,r\}\cap\{u,v,w\}|=1$ and $|\{p,r\}\cap\{u,v,w\}|=2$, respectively.

Let us focus on $o^{f|g,h}_1$. In these triples either $p\in\{u,v,w\}$ or $r\in\{u,v,w\}$, but not both (recall that $p<r$).
Since $u$, $v$, and $w$ are symmetric, we can assume that either $p=w$ or $r=w$.
In other words either $A=\{p,r,s,t\}$, $B=\{p,s,u,v\}$ and $C=\{p,t,u,v\}$ or $A=\{p,r,s,t\}$, $B=\{r,s,u,v\}$ and $C=\{r,t,u,v\}$.
Equivalently, we can drop the assumption $p<r$ and just consider all triples $A=\{p,r,s,t\}$, $B=\{p,s,u,v\}$ and $C=\{p,t,u,v\}$. 
It follows that

\[o^{f|g,h}_1 = \sum_{1\leq p,r,s,t\leq n}
\sum_{\substack{1\leq u<v\leq n\\\{p,r\}\cap\{u,v\}=\emptyset}}
f\bigl(\{p,r,s,t\}\bigr)
g\bigl(\{p,s,u,v\}\bigr)
h\bigl(\{p,t,u,v\}\bigr).\]

Observe that $o^{f|g,h}_1$ is equal to the contribution of Type II triples $(A,B,C)$ such that the unique degree $1$ vertex is in $A$, i.e.\ 
$o^{f|g,h}_1 = z_{2,\mathrm{II}}^{f|g,h}$, and we know how to compute this value in time $O(n^{\omega+2})$.

Now consider $o^{f|g,h}_2$. In these triples $\{p,r\}\subseteq\{u,v,w\}$.
Since $u$, $v$ and $w$ are symmetric, we can assume $\{p,r\}=\{v,w\}$.
It means that we count triples of the form $A=\{p,r,s,t\}$, $B=\{p,r,s,u\}$, and $C=\{p,r,t,u\}$, for every five different vertices $p<r$ and $s,t,u$.
It follows that $o^{f|g,h}_2=x_{2,\mathrm{I}}$, which is computable in $O(n^{\omega+2})$ time.

Finally we focus on Type IV triples, where there are two degree $1$ vertices $p$ and $r$, the remaining vertices are of degree 2 and $p$ and $r$ are in different hyperedges. Let $z_{2,\mathrm{IV}}^{f,g|h}$ denote the contribution to $x_2$ of Type IV triples where $p,r\in A\cup B$.
Note that then $z_{2,\mathrm{IV}}^{f,h|g}$ and $z_{2,\mathrm{IV}}^{g,h|f}$ are the contributions of Type IV triples where $p,r\in A\cup C$ and $p,r\in B\cup C$, respectively, and hence 
\begin{equation}
\label{eq:4-uniform-type-IV}
 x_{2,\mathrm{IV}}=z_{2,\mathrm{IV}}^{f,g|h}+z_{2,\mathrm{IV}}^{f,h|g}+z_{2,\mathrm{IV}}^{g,h|f}.
\end{equation}
We focus on $z_{2,\mathrm{IV}}^{f,g|h}$, i.e.\ we can assume $p\in A$ and $r\in B$. 
Then $C=\{s,t,u,v\}$ for some four different vertices $s$, $t$, $u$, $v$ diffrent from $p$ and $r$. Since all vertices except for $p$ and $r$ are of degree 2 from the handshaking lemma the number of vertices $k$ in the hypergraph satisfies $1\cdot 2+2\cdot (k-2)=4\cdot 3$, and hence $k=7$, i.e.\ there is exactly one vertex more, call it $w$.
Since $w$ is of degree 2, $w\in A\cap B$. Since all vertices in $C$ are of degree 2, two of them, say $s,t$ are in $A$, and the other two, say $u,v$ in $B$. Then,

\[
\begin{split}
z_{2,\mathrm{IV}}^{f,g|h}
=& 
\sum_{1\leq p\ne r\leq n}
\sum_{1\leq w\leq n}
\sum_{\substack{1\leq s<t\leq n\\r\not\in\{s,t\}}}
\sum_{\substack{1\leq u<v\leq n\\p\not\in\{u,v\}}}
f\bigl(\{p,w,s,t\}\bigr)
g\bigl(\{r,w,u,v\}\bigr)
h\bigl(\{s,t,u,v\}\bigr)\\
=&
\underbrace{\sum_{1\leq p,w\leq n}
\sum_{\substack{1\leq u<v\leq n\\p\not\in\{u,v\}}}
\biggl(\sum_{1\leq s<t\leq n}
f\bigl(\{p,w,s,t\}\bigr)
h\bigl(\{s,t,u,v\}\bigr)
\biggr)
\sum_{\substack{1\le r\le n\\ r\ne p}}g\bigl(r,w,u,v\bigr)}_{(**)} \ -\\
& 
\underbrace{
\sum_{1\leq p\ne r\leq n}
\sum_{1\leq w\leq n}
\sum_{\substack{1\leq s<t\leq n\\r\in\{s,t\}}}
\sum_{\substack{1\leq u<v\leq n\\p\not\in\{u,v\}}}
f\bigl(\{p,w,s,t\}\bigr)
g\bigl(\{r,w,u,v\}\bigr)
h\bigl(\{s,t,u,v\}\bigr)}_{o^{f,g|h}}\,,
\end{split}
\]
where $o^{f,g|h}$ is an overcount which we specify in a moment.
Note that by similar arguments as before, it is sufficient to assume only $\{p,r\}\cap\{u,v,w\}=\emptyset$.
Observe also that the sum $s_{p,w,u,v}=\sum_{1\leq s<t\leq n}
f\bigl(\{p,w,s,t\}\bigr)
h\bigl(\{s,t,u,v\}\bigr)$ can be evaluated for each $p,w,u,v$ 
by multiplying two $n^2\times n^2$ matrices in $O(n^{2\omega})$ arithmetic operations.
Moreover, for every $u,v,w$ we compute and store the sum $s_{u,v,w}=\sum_{1\le r\le n}g\bigl(r,u,v,w\bigr)$; this takes time $O(n^4)\subseteq O(n^{2\omega})$ . 
By noting that $\sum_{\substack{1\le r\le n\\ r\ne p}}g\bigl(r,u,v,w\bigr)=s_{u,v,w}-g\bigl(p,u,v,w\bigr)$ it follows that once all the values of $s_{p,w,u,v}$ and $s_{u,v,w}$ are computed the sum $(**)$ can be evaluated in $O(n^4)\subseteq O(n^{2+\omega})$ time.
It remains to compute the value of overcount $o^{f,g|h}$ efficiently.

Observe that $o^{f,g|h}$ is equal to the co ntribution of Type IV triples considered above where $r\in\{s,t\}$, i.e.\ the triples $(A,B,C)$ such that $A=\{p,r,s,w\}$, $B=\{r,w,u,v\}$ and $C=\{r,s,u,v\}$. These are exactly the Type II triples such that the degree $1$ vertex is in $A$.
Hence, $o^{f,g|h}=z_{2,\mathrm{II}}^{f|g,h}$, and we know how to compute this value in time $O(n^{\omega+2})$.
This finishes the proof of Theorem~\ref{thm:q234}.
\end{proof}

\section{Counting thin subgraphs in three parts}
\label{sect:homomorphims-in-three-parts}

This section proves Theorem~\ref{thm:subgraph} by relying on the techniques 
in \S4 of Fomin~{\em et al.}~\cite{FLRRS12} and invoking our 
Theorem~\ref{thm:main} as a subroutine that enables 
fast counting of injective homomorphisms in three parts.
Whereas Fomin~{\em et al.}~\cite{FLRRS12}
use the path decomposition to split $P$ into two halves of size roughly 
$k/2$ joined by a separator of size at most $p$, we split $P$ into 
a sequence of three parts of size roughly $k/3$ joined by two separators 
of size at most $p$. Accordingly, the following lemma is an immediate analog 
of Proposition 2 in Fomin~{\em et al.}~\cite{FLRRS12} (cf.~\cite{K92}).

\begin{Lem}
Let $P$ be a graph with $k$ vertices and pathwidth $p$. Then, we can partition
the vertices of $P$ into five pairwise disjoint sets $L,S,M,T,R$ such that
(i) $|L|,|M|,|R|\leq k/3$, (ii) $|S|,|T|\leq p$, and (iii) every edge of $P$
joins vertices in one or two consecutive sets in the sequence $L,S,M,T,R$. 
\end{Lem}

Imitating the design in \S4 of Fomin~{\em et al.}~\cite{FLRRS12},
we now iterate over all possible $O(n^{2p})$
guesses $\varphi$ how an injective homomorphism from $P$ to $H$ can map 
the elements of the disjoint sets $S$ and $T$ to $V(H)$. For each such 
guess $\varphi$, we use the algorithm in Lemma 2 
of Fomin~{\em et al.}~\cite{FLRRS12} to compute 
for each $A,B,C\subseteq V(H)\setminus (\varphi(S)\cup \varphi(T))$ 
with $|A|,|B|,|C|\leq k/3$ the following three quantities:
(a) the number $f_\varphi(A)$ of injective homomorphisms 
from $P[L\cup S]$ to $H[A\cup \varphi(S)]$ that extend $\varphi$, 
(b) the number $g_\varphi(B)$ of injective homomorphisms 
from $P[S\cup M\cup T]$ to $H[\varphi(S)\cup B\cup\varphi(T)]$ that extend $\varphi$, 
and
(c) the number $h_\varphi(C)$ of injective homomorphisms 
from $P[T\cup R]$ to $H[\varphi(T)\cup C]$ that extend $\varphi$.
This takes $O\bigl(n^{k/3+3p+c}\bigr)$ time for a constant $c\geq 0$ independent
of the constants $k$ and $p$; in particular the running-time bottleneck 
occurs with the functions $g_\varphi$ where we run an $O(n^{p+c})$-time
algorithm of D{\'{\i}}az~{\em et al.}~\cite{DST02} to compute the number
of homomorphisms from $P[S\cup M\cup T]$ to $H[\varphi(S)\cup B\cup\varphi(T)]$
that extend $\varphi$ for each of the $O(n^{2p+k/3})$ possibilities for 
$\varphi(S)\cup B\cup\varphi(T)$.
 
Using the algorithm in Theorem~\ref{thm:main} for each guess $\varphi$, 
we obtain the number of injective homomorphisms from $P$ to $H$ as 
$\sum_{\varphi} \Delta\bigl(f_\varphi,g_\varphi,h_\varphi\bigr)$ in time
$O\bigl(n^{(\frac{1}{2}-\tau)k+2p+c}\bigr)$. Dividing by the number
of automorphisms of $P$, we obtain the number of subgraphs isomorphic 
to $P$ in $H$ (cf.~\cite[Theorem~2]{FLRRS12}). This completes the proof of 
Theorem~\ref{thm:subgraph}.

\section{Counting set packings in three parts}
\label{sect:packings-in-three-parts}

This section proves Theorem~\ref{thm:packings}. 
Let $U$ be the $n$-element universe and let 
$\mathcal{F}\subseteq\binom{U}{s}$ be a set of $s$-element subsets of $U$
given as input. A further input is the integer $t\equiv 0\pmod 3$.
Our task is to count the number of 
$t$-tuples $(S_1,S_2,\ldots,S_t)\in\mathcal{F}^t$ that
are pairwise disjoint, that is, $S_i\cap S_j$ holds for all 
$1\leq i<j\leq t$.

The structure of the proof is to rely on standard dynamic programming 
techniques to execute the count for pairwise disjoint $t/3$-tuples, and 
then invoke the weighted disjoint triples algorithm (Theorem~\ref{thm:main}) 
to arrive at the desired count. Let us now proceed with the details.

We start by defining a sequence of functions 
$f_\ell:\binom{U}{\ell s}\rightarrow\mathbb{Z}$ that we will
then compute using dynamic programming.
For $\ell=1,2,\ldots,t/3$ and all $X\in\binom{U}{\ell s}$, let
$f_\ell(X)$ be the number of $\ell$-tuples 
$(S_1,S_2,\ldots,S_\ell)\in\mathcal{F}^\ell$ that (a) are pairwise disjoint,
and (b) satisfy $X=S_1\cup S_2\cup\cdots\cup S_\ell$.

To set up a base case for the dynamic programming, we observe 
that $f_1:\binom{U}{s}\rightarrow\mathbb{Z}$ is the indicator function 
for the subsets in $\mathcal{F}$. That is, $f_1(X)=1$ if and only if
$X\in\mathcal{F}$, and $f_1(X)=0$ otherwise. Since 
$|\mathcal{F}|\leq\binom{n}{s}$ and $s$ is a constant independent of $n$, 
we have that $f_1$ can be computed in time $O(n^{s+c})$.
Next, suppose that we have computed $f_{\ell-1}$ and want to
compute $f_{\ell}$. For each $X\in\binom{U}{\ell s}$, we use the 
following recurrence:
\[
f_\ell(X)=\sum_{Y\in\binom{X}{s}\cap\mathcal{F}}f_{\ell-1}(X\setminus Y)\,.
\]
To see that the recurrence is correct, observe that 
for every $(S_1,S_2,\ldots,S_\ell)\in\mathcal{F}^\ell$
that is pairwise disjoint with $X=S_1\cup S_2\cup\cdots\cup S_\ell$
there is a unique $Y\in\mathcal{F}\cap\binom{X}{s}$ such that
$(S_1,S_2,\ldots,S_{\ell-1})\in\mathcal{F}^{\ell-1}$
is pairwise disjoint with 
$X\setminus Y=S_1\cup S_2\cup\cdots\cup S_{\ell-1}$, namely $Y=S_{\ell}$.
In particular, the left-hand side and the right-hand side of the 
recurrence count the same $\ell$-tuples. 

To obtain the running time of the recurrence, observe that we
iterate over all $X\in\binom{U}{\ell s}$ and then over all
$Y\in\binom{X}{s}$, checking for each $Y$ whether $Y\in\mathcal{F}$.
Since both $s$ and $t$ are constants independent of $n$, 
also $\ell$ is a constant independent of $n$, and the running 
time bound becomes $O(n^{s\ell+c}(s\ell)^s)=O(n^{s\ell+c})$.
In particular, to compute the function $f_{t/3}$ using the recurrence 
thus takes $O(n^{st/3+c})$ time.

We will now apply Theorem~\ref{thm:main}. Let us take
$q=st/3$ and compute $\Delta(f_{t/3},f_{t/3},f_{t/3})$ using
the algorithm in Theorem~\ref{thm:main}. This will take
$O\bigl(n^{3q(\frac{1}{2}-\tau)+c}\bigl)=
O\bigl(n^{(\frac{1}{2}-\tau)st+c}\bigl)$, which is exactly
the claimed running time. 

To complete the proof of Theorem~\ref{thm:packings}, we observe 
that $\Delta(f_{t/3},f_{t/3},f_{t/3})$
is exactly the number of $t$-tuples of pairwise disjoint subsets
from $\mathcal{F}$, multiplied by the multinomial coefficient
$\binom{t}{t/3,t/3,t/3}$. Indeed, each $t$-tuple 
$(S_1,S_2,\ldots,S_t)\in\mathcal{F}^t$ of pairwise disjoint sets 
is counted in $\Delta(f_{t/3},f_{t/3},f_{t/3})$ exactly 
$\binom{t}{t/3,t/3,t/3}$ times, once for each possible way of 
partitioning the index set $\{1,2,\ldots,t\}$ into a three-tuple 
$(I,J,K)$ with $|I|=|J|=|K|=t/3$ such that 
$A=\cup_{\ell\in I}S_\ell$, $B=\cup_{\ell\in J}S_\ell$, 
and $C=\cup_{\ell\in K}S_\ell$; cf.~\eqref{eq:main}. 
Thus, $\binom{t}{t/3,t/3,t/3}^{-1}\Delta(f_{t/3},f_{t/3},f_{t/3})$
is the count we want. This completes the proof of Theorem~\ref{thm:packings}.

\section{On the hardness of counting in disjoint parts}
\label{sect:lower-bounds}

This section presents two results that provide partial
justification why there was an apparent barrier 
at ``meet-in-the-middle time'' for counting in disjoint parts. 

First, in the case of two disjoint parts, the problem appears to
contain no algebraic dependency that one could expoit towards
faster algorithms beyond those already presented in 
Bj\"orklund {\em et al.}~\cite{BHKK08,BHKK09}. Indeed, we can
provide some support towards this intuition by recalling that
the associated 2-tensor has full rank over the rationals
(Lemma~\ref{lem:disjmat}).

Second, in the case of three disjoint parts, we have already
witnessed (in the proof of Theorem~\ref{thm:main}) that the
associated 3-tensor does not have full rank, in essence
because the 3-tensor for matrix multiplication does
not have full rank. This prompts the question whether it was
{\em necessary} to rely on fast matrix multiplication 
to break the barrier. We can provide some support towards an 
affirmative answer by showing that any {\em trilinear} algorithm 
for $\Delta$ that breaks the barrier implies a nontrivial algorithm 
for matrix multiplication (Theorem~\ref{thm:omega-tau}).

\subsection{Disjoint pairs}
It will be convenient to use Iverson's bracket notation;
for a logical proposition $P$ we have $[P]=1$ if $P$ is true and $[P]=0$ if
$P$ is false. The $(n,k)$-{\em disjointness matrix} is the 
$\binom{n}{k}\times\binom{n}{k}$ matrix with entries $[A\cap B=\emptyset]$
for all $A,B\in\binom{U}{k}$, $0\leq k\leq n$.
The following lemma is well known.

\begin{Lem}
\label{lem:disjmat}
The $(n,k)$-disjointness matrix has full rank over the rationals.
\end{Lem}

\begin{proof}
It suffices to show that the matrix is invertible. Observe that
\[
\sum_{B\in\binom{U}{k}} [A\cap B=\emptyset] z_{|B\cap C|} = [A=C]
\]
holds for all $A,C\in\binom{U}{k}$ when the values $z_j$ 
for $j=0,1,\ldots,k$ are the solutions to the $(k+1)\times(k+1)$ 
linear system with equations
\[
\sum_{j=0}^k \binom{i}{j}\binom{n-k-i}{k-j} z_j = [i=0],\qquad i=0,1,\ldots,k\,.
\]
The coefficient matrix of this system is a lower triangular matrix
with nonzero entries on the diagonal. Thus, the system is invertible.
\end{proof}

\subsection{Fast disjoint triples implies fast matrix multiplication}

Let us prove Theorem~\ref{thm:omega-tau}.
For every $q$ 
we have by assumption a trilinear algorithm of rank $r$ for 
computing $\Delta(f,g,h)$ for inputs 
$f,g,h:\binom{U}{q}\rightarrow\mathbb{Z}$
over an $n$-element universe $U$.
That is, for every $q$ and $n$ there exist coefficients 
$\lambda_s,\alpha_A,\beta_B,\gamma_C\in\mathbb{Z}$ 
such that
\[
\Delta(f,g,h)=\sum_{s=1}^r \lambda_sF_sG_sH_s
\]
with
\[
F_s=\sum_{A\in\binom{U}{q}}\alpha_Af(A),\quad
G_s=\sum_{B\in\binom{U}{q}}\beta_Bg(B),\quad
H_s=\sum_{C\in\binom{U}{q}}\gamma_Ch(C)\,.
\]

Let us now derive from these trilinear algorithms
bilinear algorithms for matrix multiplication.
Let us recall from \eqref{eq:xj} the indeterminates $x_j$ 
with $j\equiv q\pmod 2$. In particular, a trilinear algorithm
for $\Delta=\Delta(f,g,h)$ is precisely a trilinear algorithm for $x_{3q}$.

Our proof strategy is to derive a bilinear algorithm for matrix
multiplication from a trilinear algorithm for $x_0$, and then
derive a trilinear algorithm for $x_0$ from the first family 
of equations (\S\ref{sect:system}) and the low-rank algorithm 
for $\Delta$. Finally, we use recursion on the bilinear algorithm
to conclude that $\omega\leq 3-\tau$.

Let $P$ and $Q$ be matrices of size $N\times N$. 
Without loss of generality we may assume that $q$ is even and 
that $n$ is divisible by 3, with $N=\binom{n/3}{q/2}=\Omega(n^{q/2})$. 

Partition $U$ into three disjoint sets $U_1,U_2,U_3$ of size $n/3$.
Define the function $f_P:\binom{U}{q}\rightarrow R$ for 
$K_1\in\binom{U_1}{q/2}$ and $K_2\in\binom{U_2}{q/2}$
by setting $f_P(K_1\cup K_2)=P(K_1,K_2)$, and let $f_P$ vanish
elsewhere.
Similarly, define the function $g_Q:\binom{U}{q}\rightarrow R$ for 
$K_2\in\binom{U_2}{q/2}$ and $K_3\in\binom{U_3}{q/2}$
by setting $g_Q(K_2\cup K_3)=Q(K_2,K_3)$, and let $g_Q$ vanish
elsewhere.

Assume that we have a trilinear algorithm of rank $r$ over the integers 
that computes $x_0$. We claim that we can transform this trilinear 
algorithm into a bilinear algorithm of rank $r$ that multiplies $P$ and $Q$. 
Indeed (cf.~\cite[\S9]{Pan1984}), we can fix $f=f_P$ and $g=g_Q$ in 
the trilinear equations, and for each $K_1\in\binom{U_1}{q/2}$ and 
$K_3\in\binom{U_3}{q/2}$ solve for the indeterminate $h(C)$ with 
$C=K_1\cup K_3$ to determine the $(K_1,K_3)$-entry of the product 
matrix $PQ$.

Recalling the first family of linear equations from \S\ref{sect:system},
from the proof of Lemma~\ref{lem:yi} and the structure of 
equations \eqref{eq:yi-tp} and \eqref{eq:tpfgh} we can observe 
that the right-hand side $y_i$ of the first family has trilinear
rank at most 
\begin{equation}
\label{eq:yi-trilinear}
O\bigl(n^{(3/2-\gamma)q+c}\bigr)
\end{equation}
whenever $0\leq i\leq (3/2-\gamma)q$.
Thus, using the first family we can show that the trilinear rank 
of $x_0$ is small by showing that the indeterminates $x_j$ for 
large values of $j$ have low trilinear rank.

Towards this end, 
let us solve for $x_{j}$ with $j\geq 3q-2d$ using the trilinear algorithm 
for $\Delta$ as a subroutine. 
That is, we iterate over all possible choices
for the intersecting part of a triple $(A,B,C)$ with $|A\oplus B\oplus C|=j$,
and for each such choice use $\Delta$ to sum over the disjoint parts 
to accumulate $x_j$.
Let us now make this more precise.
For sets $X,Y\subseteq U$ let us abbreviate $XY=X\cap Y$ and
$\bar X=U\setminus X$. 
For $A,B,C\in\binom{U}{q}$, the {\em intersecting part}
of the triple $(A,B,C)$ is the tuple
of disjoint sets 
\begin{equation}
\label{eq:int-def}
I(A,B,C)=(AB\bar C,A\bar BC,\bar ABC,ABC)\,.
\end{equation}
The {\em size} of $I(A,B,C)$ is 
\begin{equation}
\label{eq:int-size}
\begin{split}
|I(A,B,C)|&=
|AB\bar C|
+|A\bar BC|
+|\bar ABC|
+|ABC|\\
&=|AB|+|AC|+|BC|-2|ABC|\,.
\end{split}
\end{equation}
We have that $j=|A\oplus B\oplus C|$ is large if and only if 
$|I(A,B,C)|$ is small. In more precise terms, from \eqref{eq:int-size}
we have
\begin{equation}
\label{eq:j-int}
\begin{split}
j&=|A\oplus B\oplus C|\\
&=|A|+|B|+|C|
-2|AB|-2|AC|-2|BC|
+4|ABC|\\
&=3q-2|I(A,B,C)|\,.
\end{split}
\end{equation}
The {\em disjoint part} of $(A,B,C)$ is the tuple 
of disjoint sets
\begin{equation}
\label{eq:disj-def}
D(A,B,C)=(A\bar B\bar C,\bar AB\bar C,\bar A\bar BC)\,.
\end{equation}
It is immediate that $I(A,B,C)$ and $D(A,B,C)$ together
uniquely determine the triple $(A,B,C)$.

Now consider an arbitrary triple $(A,B,C)$ with 
$j=|A\oplus B\oplus C|\geq 3q-2d$. We know that there is a 
unique quadruple $(I_1,I_2,I_3,I_4)$ of disjoint sets 
with $I(A,B,C)=(I_1,I_2,I_3,I_4)$. Furthermore, from \eqref{eq:j-int}
and $j\geq 3q-2d$ it follows that $|I_1|+|I_2|+|I_3|+|I_4|\leq d$.
Thus, it suffices to iterate over at most $4^d(d+1)n^{d}$ quadruples 
$(I_1,I_2,I_3,I_4)$ to match the intersecting part of $(A,B,C)$.

So suppose we have fixed $(I_1,I_2,I_3,I_4)$ with 
$|I_1|+|I_2|+|I_3|+|I_4|\leq d$, and let $I=I_1\cup I_2\cup I_3\cup I_4$.
We can now capture each $(A,B,C)$ with $I(A,B,C)=(I_1,I_2,I_3,I_4)$ 
by means of the disjoint part $D(A,B,C)$. That is, there exists a unique 
disjoint triple $(D_1,D_2,D_3)$ of subsets $D_1,D_2,D_3\subseteq\bar I$ 
such that $D(A,B,C)=(D_1,D_2,D_3)$.
Furthermore, from \eqref{eq:int-def} and \eqref{eq:disj-def} it is 
immediate that
\[
\begin{split}
|D_1|&=q-|I_1|-|I_2|-|I_4|\,,\\
|D_2|&=q-|I_1|-|I_3|-|I_4|\,,\\
|D_3|&=q-|I_2|-|I_3|-|I_4|\,.
\end{split}
\]

Since the sets $D_1,D_2,D_3$ in general have size different from
$q$, let us introduce the following padding to obtain a valid input
for $\Delta$. Let $E_1,E_2,E_3$ be three sets that are disjoint from each
other and $U$, with 
$|E_1|=|I_1|+|I_2|+|I_4|$,
$|E_2|=|I_1|+|I_3|+|I_4|$, and
$|E_3|=|I_2|+|I_3|+|I_4|$.
Let $U'=U\cup E_1\cup E_2\cup E_3$ and $n'=n+|E_1|+|E_2|+|E_3|\leq n+3q$.

We are now ready to construct an input for $\Delta$.
Define three functions $f',g',h':\binom{U'}{q}\rightarrow R$ for all 
$A',B',C'\in\binom{U'}{q}$ by 
\[
\begin{split}
f'(A')&=
\begin{cases}
f\bigl((A'\cap U)\cup I_1\cup I_2\cup I_4\bigr) & \text{if $A'\cap U\subseteq\bar I$ and $A'\cap (U'\setminus U)=E_1$;}\\
0           & \text{otherwise.}
\end{cases}\\
g'(B')&=
\begin{cases}
g\bigl((B'\cap U)\cup I_1\cup I_3\cup I_4\bigr) & \text{if $B'\cap U\subseteq\bar I$ and $B'\cap (U'\setminus U)=E_2$;}\\
0           & \text{otherwise.}
\end{cases}\\
g'(C')&=
\begin{cases}
g\bigl((C'\cap U)\cup I_2\cup I_3\cup I_4\bigr) & \text{if $C'\cap U\subseteq\bar I$ and $C'\cap (U'\setminus U)=E_3$;}\\
0           & \text{otherwise.}
\end{cases}
\end{split}
\]
By construction, we now have that to every triple $(A,B,C)$ with 
$I(A,B,C)=(I_1,I_2,I_3,I_4)$ there corresponds a unique disjoint triple 
$(A',B',C')$ such that 
\[
f(A)g(B)h(C)=f'(A')g'(B')h'(C')\,. 
\]
Taking the sum over all $(I_1,I_2,I_3,I_4)$, we have
\[
x_j = 
\sum_{\substack{(I_1,I_2,I_3,I_4)\\|I_1|+|I_2|+|I_3|+|I_4|=(3q-j)/2}}\Delta(f',g',h')\,.
\]
Thus, using the trilinear algorithm for $\Delta$ 
for each choice of $(I_1,I_2,I_3,I_4)$, we can compute 
$x_j$ for all $3q-2d\leq j\leq 3q$ with a trilinear algorithm of rank
\begin{equation}
\label{eq:xj-trilinear}
O\bigl(4^{d}(d+1)n^{d}(n+3q)^{3q(1/2-\tau)+c}\bigr)=O(n^{d+3q(1/2-\tau)+c})\,.
\end{equation}
Take $d=\lceil\gamma q\rceil$ so that together with equations 
from the first family \eqref{eq:yi-trilinear} we have enough equations 
to solve for $x_0$. It remains to select $\gamma$ so that 
\eqref{eq:yi-trilinear} and \eqref{eq:xj-trilinear} are balanced.
We have that \eqref{eq:yi-trilinear} and 
\eqref{eq:xj-trilinear} are balanced when
\[
(3/2-\gamma)q=\gamma q+3q(1/2-\tau)\,.
\]
That is, when $\gamma = 3\tau/2$. We thus have a trilinear algorithm 
for $x_0$ that has rank $r=O\bigl(n^{(3-\tau)q/2+c}\bigr)$ for a 
constant $c$ independent of $n$ and $q$.
That is, we have a bilinear algorithm of rank 
$r=O\bigl(n^{(3-\tau)q/2+c}\bigr)$ to multiply two $N\times N$ matrices
with $N=\Omega(n^{q/2})$. For any constant $\epsilon>0$ we can now
obtain, by selecting a large enough $q$, a bilinear algorithm of
rank $r=O(N^{3-\tau+\epsilon})$ to multiply two $N\times N$ matrices.
Taking a large enough $N$ and using recursion 
(cf.~\cite[Theorem~2.1]{Pan1984}), we conclude that 
$\omega\leq 3-\tau+\epsilon$. Since $\epsilon$ was arbitrary,
$\omega\leq 3-\tau$.

\section*{Acknowledgments}

A preliminary conference abstract of this work has appeared as A. Bj\"orklund, P. Kaski, and \L. Kowalik, ``Counting thin subgraphs via packings faster than meet-in-the-middle time,'' Proceedings of the 25th ACM-SIAM Symposium on Discrete Algorithms (SODA 2014, Portland, Oregon, January 5--7, 2014), SIAM, Philadelphia, PA, 2014, pp.~594--603.
This research was supported in part by the 
Swedish Research Council, Grant VR 2012-4730 (A.B.),
the Academy of Finland, Grants 252083, 256287, and 283437 (P.K.), and 
by t	he National Science Centre of Poland, Grants N206 567140 and 2013/09/B/ST6/03136 (\L.K.).


\bibliographystyle{abbrv}


\end{document}

%% file: type.pdf_tex
\begingroup%
  \makeatletter%
  \providecommand\color[2][]{%
    \errmessage{(Inkscape) Color is used for the text in Inkscape, but the package 'color.sty' is not loaded}%
    \renewcommand\color[2][]{}%
  }%
  \providecommand\transparent[1]{%
    \errmessage{(Inkscape) Transparency is used (non-zero) for the text in Inkscape, but the package 'transparent.sty' is not loaded}%
    \renewcommand\transparent[1]{}%
  }%
  \providecommand\rotatebox[2]{#2}%
  \ifx\svgwidth\undefined%
    \setlength{\unitlength}{180.771875bp}%
    \ifx\svgscale\undefined%
      \relax%
    \else%
      \setlength{\unitlength}{\unitlength * \real{\svgscale}}%
    \fi%
  \else%
    \setlength{\unitlength}{\svgwidth}%
  \fi%
  \global\let\svgwidth\undefined%
  \global\let\svgscale\undefined%
  \makeatother%
  \begin{picture}(1,0.82351743)%
    \put(0,0){\includegraphics[width=\unitlength]{type.pdf}}%
    \put(-0.00315487,0.73180796){\color[rgb]{0,0,0}\makebox(0,0)[lb]{\smash{$A$
}}}%
    \put(0.84653482,0.73180796){\color[rgb]{0,0,0}\makebox(0,0)[lb]{\smash{$B$
}}}%
    \put(0.66951614,0.01488227){\color[rgb]{0,0,0}\makebox(0,0)[lb]{\smash{$C$
}}}%
  \end{picture}%
\endgroup%

%% file: triple-intersect-bw.pdf_tex
\begingroup%
  \makeatletter%
  \providecommand\color[2][]{%
    \errmessage{(Inkscape) Color is used for the text in Inkscape, but the package 'color.sty' is not loaded}%
    \renewcommand\color[2][]{}%
  }%
  \providecommand\transparent[1]{%
    \errmessage{(Inkscape) Transparency is used (non-zero) for the text in Inkscape, but the package 'transparent.sty' is not loaded}%
    \renewcommand\transparent[1]{}%
  }%
  \providecommand\rotatebox[2]{#2}%
  \ifx\svgwidth\undefined%
    \setlength{\unitlength}{263.4359375bp}%
    \ifx\svgscale\undefined%
      \relax%
    \else%
      \setlength{\unitlength}{\unitlength * \real{\svgscale}}%
    \fi%
  \else%
    \setlength{\unitlength}{\svgwidth}%
  \fi%
  \global\let\svgwidth\undefined%
  \global\let\svgscale\undefined%
  \makeatother%
  \begin{picture}(1,0.61316631)%
    \put(0,0){\includegraphics[width=\unitlength]{triple-intersect-bw.pdf}}%
    \put(0.70985,0.22886132){\color[rgb]{0,0,0}\makebox(0,0)[lb]{\smash{$B$
}}}%
    \put(0.58837834,0.01021234){\color[rgb]{0,0,0}\makebox(0,0)[lb]{\smash{$C$
}}}%
    \put(0.89813107,0.53254046){\color[rgb]{0,0,0}\makebox(0,0)[lb]{\smash{$Z$
}}}%
    \put(0.13285963,0.22886132){\color[rgb]{0,0,0}\makebox(0,0)[lb]{\smash{$A$
}}}%
  \end{picture}%
\endgroup%